\documentclass[runningheads]{llncs}

\usepackage[utf8]{inputenc}
\usepackage{amsmath}
\usepackage{amssymb}
\usepackage{graphicx}
\usepackage{tabularx}
\usepackage{array}
\usepackage{pifont}
\usepackage{xspace}
\usepackage{todonotes}
\usepackage{paralist}
\usepackage{thm-restate}
\usepackage{subfigure}
\usepackage[pdfpagelabels,colorlinks,citecolor=blue,linkcolor=blue,urlcolor=blue]{hyperref}
\usepackage[english]{babel}
\usepackage{amsopn}
\usepackage{latexsym}
\usepackage{multirow}
\usepackage{multicol}
\usepackage{booktabs}
\usepackage[capitalise]{cleveref}
\usepackage{color, colortbl}
\usepackage{booktabs}

\spnewtheorem{clm}{Claim}{\bfseries}{\rmfamily}

\newcommand{\myparagraph}[1]{\smallskip\noindent\textbf{\boldmath #1}}
\newenvironment{claimproof}{\noindent{\itshape Claim Proof.}}{~\hfill
$\blacksquare$\smallskip}

\graphicspath{{}}

\definecolor{lightcyan}{rgb}{0.88,1,1}
\definecolor{antiquewhite}{rgb}{0.98, 0.92, 0.84}
\DeclareMathOperator\indeg{indeg}
\DeclareMathOperator\outdeg{outdeg}

\newcounter{casecounter}
\newcounter{subcasecounter}
\newcounter{subsubcasecounter}
\makeatletter
\newcommand{\ccase}[2][]{%
	\stepcounter{casecounter}%
	\setcounter{subcasecounter}{0}%
	\protected@write \@auxout {}{\string \newlabel {#2}{{#1\thecasecounter}{\thepage}{#1\thecasecounter}{#2}{}} }%
	\hypertarget{#2}{\noindent\textbf{Case #1\thecasecounter.}}
}

\newcommand{\subcase}[2][]{%
	\stepcounter{subcasecounter}%
	\setcounter{subsubcasecounter}{0}%
	\protected@write \@auxout {}{\string \newlabel {#2}{{#1\thecasecounter.\thesubcasecounter}{\thepage}{#1\thecasecounter.\thesubcasecounter}{#2}{}} }%
	\hypertarget{#2}{\noindent\textbf{Case #1\thecasecounter.\thesubcasecounter.}}
}

\newcommand{\subsubcase}[2][]{%
	\stepcounter{subsubcasecounter}%
	\protected@write \@auxout {}{\string \newlabel {#2}{{#1\thecasecounter.\thesubcasecounter.\thesubsubcasecounter}{\thepage}{#1\thecasecounter.\thesubcasecounter.\thesubsubcasecounter}{#2}{}} }%
	\hypertarget{#2}{\noindent\textbf{Case #1\thecasecounter.\thesubcasecounter.\thesubsubcasecounter.}}
}
\makeatother

\newcommand{\skel}{\mathrm{skel}\xspace}

\newcommand{\pisp}{independent-parallel SP-graph\xspace}
\newcommand{\pisps}{independent-parallel SP-graphs\xspace}
\newcommand{\md}{\mathrm{mid}\xspace}
\begin{document}
	% ============================================================
	\title{Spirality and Rectilinear Planarity Testing of Independent-Parallel SP-Graphs\thanks{This work is partially supported by: $(i)$ MIUR, grant 20174LF3T8 ``AHeAD: efficient Algorithms for HArnessing networked Data'', $(ii)$ Dipartimento di Ingegneria - Universit\`a degli Studi di Perugia, grants RICBA19FM: ``Modelli, algoritmi e sistemi per la visualizzazione di grafi e reti'' and RICBA20EDG: ``Algoritmi e modelli per la rappresentazione visuale di reti''.}}
	\author{Walter Didimo\inst{1},
		Michael Kaufmann\inst{2},
		Giuseppe Liotta\inst{1},
		Giacomo Ortali\inst{1}
	}%end author
	
	\date{}
	
	\institute{
		Universit\`a degli Studi di Perugia, Italy\\
		\email {\{walter.didimo,giuseppe.liotta\}@unipg.it, giacomo.ortali@studenti.unipg.it}
		\and
		University of T\"ubingen, Germany\\
		\email {mk@informatik.uni-tuebingen.de}
	}%end institute

	% ============================================================
	\maketitle
	% ============================================================
	
	%\linenumbers
	%
\begin{abstract}
	
	We study the long-standing open problem of efficiently testing rectilinear planarity of series-parallel graphs (SP-graphs) in the variable embedding setting. A key ingredient behind the design of a linear-time testing algorithm for SP-graphs of vertex-degree at most three is that one can restrict the attention to a constant number of ``rectilinear shapes'' for each series or parallel component. To formally describe these shapes the notion of spirality can be used.
	This key ingredient no longer holds for SP-graphs with vertices of degree four, as we prove a logarithmic lower bound on the spirality of their components. 
	The bound holds even for the independent-parallel SP-graphs, in which no two parallel components share a pole. Nonetheless, by studying the spirality properties of the independent-parallel SP-graphs, we are able to design
	a linear-time rectilinear planarity testing algorithm for this graph family. 

%This sheds new light on the computational complexity of the rectilinear~planarity~testing problem.

	%We also study the spirality properties of the independent-parallel SP-graphs, which are series parallel graphs where no two parallel components share a pole. We are able to design a linear-time rectilinear planarity testing for this class of graphs in the variable embedding setting, which sheds new light on the long-standing open problem about the computational complexity of rectlinear planarity testing.

%This paper addresses the long-standing open problem of efficiently testing rectilinear planarity in the variable embedding setting. We show that rectilinear planarity testing can be solved in linear time for a meaningful family of series-parallel graphs, which we call \pisp. A series-parallel graph is \pisp if no two parallel components share a pole.

\end{abstract}

\section{Introduction}\label{se:intro}
Rectilinear planarity testing asks whether a planar 4-graph (i.e., with vertex-degree at most four) admits a planar orthogonal drawing without edge bends. It is a classical subject of study in graph drawing, partly for its theoretical beauty and partly because it is at the heart of the algorithms that compute bend-minimum orthogonal drawings, which find applications in several domains (see, e.g.~\cite{DBLP:books/ph/BattistaETT99,dl-gvdm-07,DBLP:reference/crc/DuncanG13,DBLP:books/sp/Juenger04,DBLP:conf/dagstuhl/1999dg,DBLP:books/ws/NishizekiR04}). Rectilinear planarity testing is NP-complete~\cite{DBLP:journals/siamcomp/GargT01}, it belongs to the XP-class when parameterized by treewidth~\cite{DBLP:conf/gd/GiacomoLM19}, and it is FPT tractable when parameterized by the number of degree-4 vertices~\cite{DBLP:conf/isaac/DidimoL98}.  Polynomial-time solutions exist for restricted versions of the problem. Namely, if the algorithm must preserve a given planar embedding, rectilinear planarity testing  can be solved in subquadratic time for general graphs~\cite{DBLP:journals/jgaa/CornelsenK12,DBLP:conf/gd/GargT96a}, and in linear time for  planar 3-graphs~\cite{DBLP:journals/jgaa/RahmanNN03} and for biconnected series-parallel graphs (SP-graphs for short)~\cite{DBLP:conf/gd/Didimo0LO20}. When the planar embedding is not fixed, linear-time solutions exist for (families of) planar 3-graphs~\cite{DBLP:conf/soda/DidimoLOP20,DBLP:conf/cocoon/Hasan019,DBLP:journals/ieicet/RahmanEN05,DBLP:journals/siamdm/ZhouN08} and for outerplanar graphs~\cite{DBLP:conf/gd/Frati20}.
A polynomial-time solution for SP-graphs has been known for a long time~\cite{DBLP:journals/siamcomp/BattistaLV98},  but establishing whether there is a linear-time algorithm for this graph family remains a long-standing open problem~\cite{DBLP:conf/gd/BrandenburgEGKLM03}; to date, the most efficient algorithm for $n$-vertex SP-graphs has complexity $O(n^3 \log n)$~\cite{DBLP:conf/gd/GiacomoLM19}.

This paper sheds new light on this long-standing open problem. We study it along the lines that led to linear-time testing algorithms for degree-3 SP-graphs~\cite{DBLP:journals/siamdm/ZhouN08} and for general planar 3-graphs~\cite{DBLP:conf/soda/DidimoLOP20}.
%We highlight some of the difficulties that stem from the degree-4 vertices and show how to overcome them to design a linear-time algorithm for a class of degree-4 SP-graphs, called \emph{independent-parallel}, that properly includes all biconnected degree-3 SP-graphs. In a independent-parallel SP-graph no two parallel components share a pole. To better describe our contribution, we briefly recall some fundamental aspects of these previous approaches.
We highlight some of the difficulties that stem from the degree-4 vertices and show how to overcome them to design a linear-time algorithm for a class of degree-4 SP-graphs that properly includes all biconnected degree-3 SP-graphs. To better describe our contribution, we briefly recall some fundamental aspects of these previous approaches.

In a nutshell, the linear-time algorithms of~\cite{DBLP:conf/soda/DidimoLOP20,DBLP:journals/siamdm/ZhouN08} are based on recursive approaches that, at each step determine if a (series or parallel) component of the graph is rectilinear planar by suitably combining rectilinear representations of its sub-components. For both algorithms, a key ingredient to achieve linear-time complexity is that it is enough to consider a constant number of rectilinear planar representations at each composition step. Another key ingredient is that the ``shapes'' of these representations can be succinctly described in $O(1)$. In~\cite{DBLP:conf/soda/DidimoLOP20} the shape is described by using the concept of \emph{spirality}, a number that describes how much a rectilinear planar representation is ``rolled up''.  Roughly, the spirality of a representation is the number of right minus left turns in any oriented path between the poles of its corresponding component
(see Section~\ref{se:preli} and Fig.~\ref{fi:intro} for an example). It is easy to see that also the shapes considered in~\cite{DBLP:journals/siamdm/ZhouN08} can be described in terms of spirality.
Hence, the possible representations for each component are succinctly described by a set of spirality values of constant~size.

\begin{figure}[tb]
	\centering
	\subfigure[]{
		\includegraphics[height=0.205\columnwidth,page=2]{holes.pdf}
		\label{fi:intro-a}
	}
	\hfill
	\subfigure[]{
		\includegraphics[height=0.205\columnwidth,page=1]{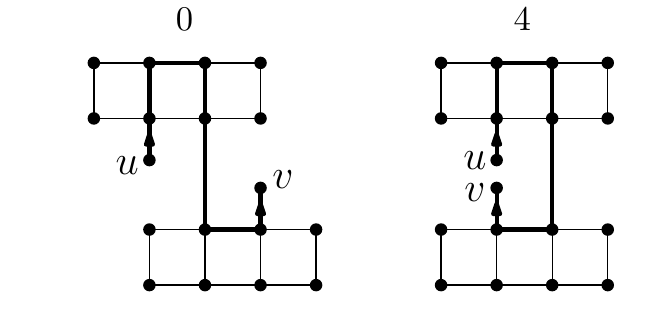}
		\label{fi:intro-b}
	}
	\caption{Two components that are: (a) rectilinear planar for spiralities 0 and 2, but not 1 (which requires a bend, shown as a cross); (b) rectilinear planar only for spiralities 0 and 4. In bold, an arbitrary path from the pole $u$ to the pole $v$.}
	\label{fi:intro}
\end{figure}

A first difficulty in extending the above approaches to degree-4 SP-graphs is that we loose one of the key ingredients: As we show in Section~\ref{sse:lowerbound}, there exist $n$-vertex SP-graphs whose rectilinear planar representations require components with $\Omega(\log n)$ spirality. For these instances a testing algorithm may need to take into account $\Omega(\log n)$ rectilinear planar representations per component. To complicate matters even further, it is not obvious how to use the spirality to construct a succinct description of these $\Omega(\log n)$ representations. Consider, for example, the component of Fig.~\ref{fi:intro-a}: It is rectilinear planar for spirality 0, it is not rectilinear planar for spirality 1, but it becomes rectilinear planar again when the spirality is 2. As another example, the component of Fig.~\ref{fi:intro-b} is rectilinear planar for spirality 0 and 4, but not for any intermediate value of spirality. The absence of regularity is an obstacle to the design of a succinct description based on whether a component is rectilinear planar for consecutive spirality values.

We study a class of 4-graphs called \emph{independent-parallel SP-graphs}, which are such that no two parallel components share a pole; see Fig.~\ref{fi:preli-g}. The component in Fig.~\ref{fi:intro-a} and the graphs used to prove the $\Omega(\log n)$ spirality lower bound (see Section~\ref{sse:lowerbound}) are independent-parallel.
%The observations above motivated us to study in Section~\ref{sse:intervals} the spirality  properties of the independent-parallel series-parallel graphs (they are SP-graphs such that no two parallel components share a pole, see Section~\ref{se:preli}).
%As it will be more clear in the rest of the paper, both the instances used to prove the lower bound of Section~\ref{sse:lowerbound} and the one of Figure~\ref{fi:-example-a} are parallel independent series-parallel graphs.
By carefully analyzing the spirality properties of \pisps (see Section~\ref{sse:intervals}), we can overcome the previously described difficulties and design a linear-time rectilinear planarity testing algorithm for this graph family (see Section~\ref{se:testing}). The algorithm uses a set of composition techniques to compute in constant time a succinct description of the rectilinear representations of each component.
%We remark that the previously known rectilinear planarity testing algorithm for SP-graphs has complexity $O(n^3 \log n)$~\cite{DBLP:conf/gd/GiacomoLM19}.
Future research directions are discussed in Section~\ref{se:final}.
For space reasons, some proofs are~in~the~appendix.

\section{Preliminaries}\label{se:preli}
\myparagraph{Orthogonal drawings and representations.} A \emph{planar orthogonal drawing}~$\Gamma$ of a planar graph $G$ is a crossing-free drawing that maps each vertex of $G$ to a distinct point of the plane and each edge of $G$ to a sequence of horizontal and vertical segments between its end-points~\cite{DBLP:books/ph/BattistaETT99,DBLP:reference/crc/DuncanG13,DBLP:books/ws/NishizekiR04}. A graph is \emph{rectilinear planar} if it admits a planar orthogonal drawing~without~bends. An \emph{orthogonal representation}~$H$ describes the shape of a class of orthogonal drawings in terms of sequences of bends along the edges and angles at the vertices. A drawing $\Gamma$ of~$H$ can be computed in linear time~\cite{DBLP:journals/siamcomp/Tamassia87}. If~$H$ has no bend, it is a \emph{rectilinear representation}.

\myparagraph{SP-graphs and SPQ-trees.} Let $G$ be a biconnected graph. The \emph{SPQR-tree}~$T$ of $G$ describes the decomposition of $G$ into its triconnected components, and can be computed in linear time~\cite{DBLP:books/ph/BattistaETT99,DBLP:conf/gd/GutwengerM00,DBLP:journals/siamcomp/HopcroftT73}. If every triconnected component of~$G$ is not a triconnected graph, $G$ is a \emph{series-parallel graph}, or \emph{SP-graph} for short. In this case $T$ is simply called SPQ-tree and contains three types of nodes: \emph{S-nodes}, \emph{P-nodes}, and \emph{Q-nodes}. The degree-1 nodes of $T$ are Q-nodes, each corresponding to a distinct edge of $G$.
%while S- and P-nodes represent series- and parallel-components, respectively.
If $\nu$ is an S-node (resp. a P-node) it represents a series-component (resp. parallel-component), denoted as $\skel(\nu)$ and called the \emph{skeleton} of $\nu$.
%If~$\nu$ is an S- or a P-node, the corresponding triconnected component is the \emph{skeleton} of $\nu$ and denoted as $\skel(\nu)$.
If $\nu$ is an S-node, $\skel(\nu)$ is a simple cycle of length at least three; if $\nu$ is a P-node, $\skel(\nu)$ is a bundle of at least three multiple edges. Any two S-nodes (resp. P-nodes) are never adjacent in~$T$.
A \emph{real edge} (resp. \emph{virtual edge}) in $\skel(\nu)$ corresponds to a Q-node (resp. an S- or a P-node) adjacent to $\nu$ in~$T$.
%A \emph{virtual edge} in $\skel(\nu)$ corresponds to an S- or a P-node adjacent to~$\nu$ in~$T$.

\myparagraph{SPQ$^*$-trees}. Testing whether a simple cycle is rectilinear planar is trivial (if and only if it has at least four vertices). Hence, we shall assume that $G$ is a biconnected SP-graph different from a simple cycle and we use a variant of the SPQ-tree called \emph{SPQ$^*$-tree} (refer to Fig.~\ref{fi:preli-graph}).
%In this variant, each degree-1 node of $T$ is a \emph{Q$^*$-node} instead of a Q-node, and represents a maximal path of vertices $u_1, u_2, \dots, u_k$ in $G$, such that $\deg(u_i)=2$, for $2 \leq i \leq k-1$, while $\deg(u_1)\geq 3$ and $\deg(u_k) \geq 3$. In other words, each Q$^*$-node represents a maximal chain of edges (possibly a single edge) starting and ending with vertices of degree larger than two and passing through a sequence of degree-2 vertices (possibly none).
In an SPQ$^*$-tree, each degree-1 node of $T$ is a \emph{Q$^*$-node}, and represents a maximal chain of edges of $G$ (possibly a single edge) starting and ending at vertices of degree larger than two and passing through a sequence of degree-2 vertices only (possibly none). If~$\nu$ is an S- or a P-node, an edge of $\skel(\nu)$ corresponding to a Q$^*$-node $\mu$ is virtual if $\mu$ is a chain of at least two edges, else it is a real edge.
%is still a real edge if $\mu$ represents a single edge, while it is a virtual edge if $\mu$ is a chain of at least two edges.
%
\begin{figure}[tb]
	\centering
	\subfigure[$G$]{
		\includegraphics[height=0.3\columnwidth,page=1]{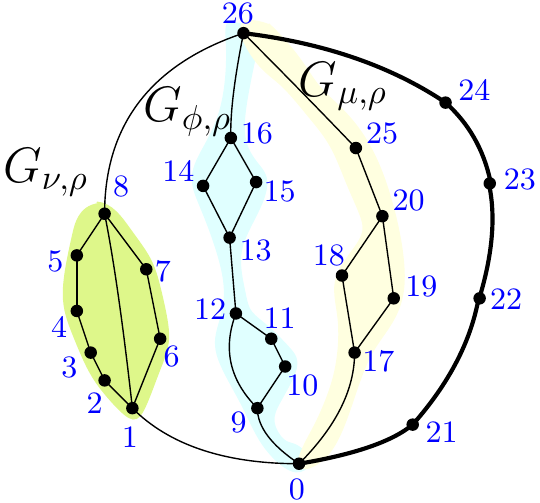}
		\label{fi:preli-g}
	}
	\hfill
	\subfigure[$H$]{
		\includegraphics[height=0.3\columnwidth,page=2]{preli-graph.pdf}
		\label{fi:preli-h}
	}
	\hfill
	\subfigure[$T_\rho$]{
		\includegraphics[height=0.4\columnwidth,page=3]{preli-graph.pdf}
		\label{fi:preli-sprq}
	}
	\caption{(a) An (independent-parallel) SP-graph $G$. (b)  A rectilinear representation $H$ of $G$. (c) The SPQ$^*$-tree $T_\rho$ of $G$, where $\rho$ represents the thick chain; Q$^*$-nodes are small squares; the left-to-right order of the children of each P-node reflects the embedding of~$H$. The components of the nodes $\nu$, $\mu$, $\phi$ are shown, as well as their skeletons: virtual edges are dashed and the reference is thicker.
%(a) An (independent-parallel) SP-graph $G$ with a given embedding. (b)  A rectilinear representation $H$ of $G$ with a slightly different embedding. (c) The SPQ$^*$-tree $T_\rho$ of $G$, where $\rho$ represents the chain $[0,21,22,23,24,26]$; Q$^*$-nodes are small squares; the left-to-right order of the children of each P-node reflects the embedding of~$H$. The components of the nodes $\nu$, $\mu$, $\phi$ are shown, as well as their skeletons: virtual edges are dashed and the reference is thicker.
}
	\label{fi:preli-graph}
\end{figure}
%
%As for SPQ-trees, an SPQ$^*$-tree $T$ of $G$ can be used to deal with the whole set of planar embeddings of $G$ by rooting $T$ at every possible Q$^*$-node.
For any given Q$^*$-node $\rho$ of $T$, denote by $T_\rho$ the tree $T$ rooted at $\rho$.
The chain of edges represented by $\rho$ is the \emph{reference chain} of $G$ with respect to $T_\rho$. If $\nu$ is an S- or a P-node distinct from the root child of~$T_\rho$, then $\skel(\nu)$ contains a virtual edge that has a counterpart in the skeleton of its parent; this edge is the \emph{reference edge} of~$\skel(\nu)$. If $\nu$ is the root child, the \emph{reference edge} of $\skel(\nu)$ is the edge corresponding to $\rho$.
%it is either a real edge, if the reference chain is a single edge, or a virtual edge, if the reference chain has at least two edges.
For any S- or P-node $\nu$ of $T_\rho$, the end-vertices of the reference edge of $\skel(\nu)$ are the \emph{poles} of $\nu$ and of $\skel(\nu)$. We remark that $\skel(\nu)$ does not change if we change~$\rho$. However, if $\nu$ is an S-node, its poles depend on~$\rho$; namely, if $\rho'$ is a  Q$^*$-node in the subtree of $T_\rho$ rooted at~$\nu$, the poles of~$\nu$ in $T_{\rho'}$ are different from those in $T_\rho$. Conversely, the poles of a P-node stay the same independent of the root of~$T$. For a Q$^*$-node $\nu$ of~$T_\rho$ (including $\rho$), the \emph{poles} of~$\nu$ are the end-vertices of the corresponding chain, and do not change when the root of~$T$ changes.
For any S- or P-node $\nu$ of~$T_\rho$, the \emph{pertinent graph} $G_{\nu,\rho}$ of $\nu$ is the subgraph of $G$ formed by the union of the chains represented by the leaves in the subtree of $T_\rho$ rooted at $\nu$. The \emph{poles} of~$G_{\nu,\rho}$ are the poles of~$\nu$. The \emph{pertinent graph} of a Q$^*$-node $\nu$ (including the root) is the chain represented by $\nu$, and its \emph{poles} are the poles of $\nu$.
%If $\nu$ is a Q$^*$-node (possibly $\nu = \rho$), the \emph{pertinent graph} $G_{\nu,\rho}$ of~$\nu$ is the chain represented by~$\nu$, and again its poles are the poles of~$\nu$.
%(the pertinent graph of a Q$^*$-node $\nu$ does not depend on the choice of $\rho$, but we denote it as $G_{\nu,\rho}$ for homogeneity of notation).
Any graph $G_{\nu,\rho}$ is also called a \emph{component} of~$G$ (with respect to~$\rho$). If $\mu$ is a child of $\nu$, we call $G_{\mu,\rho}$ a \emph{child component}~of~$\nu$.
If $H$ is a rectilinear representation of $G$, for any node $\nu$ of $T_\rho$, the restriction $H_{\nu,\rho}$ of $H$ to $G_{\nu,\rho}$ is a \emph{component} of $H$ (with respect to~$\rho$).
%
%%% EMBEDDINGS
%Tree $T_\rho$ is used to implicitly describe all planar embeddings of $G$ that have the chain represented by $\rho$ on the external face. We call such a chain, the \emph{reference chain} of $G$ when $T$ is rooted at $\rho$. Notice that, in every planar embedding of $G$, all the edges in the chain represented by a Q$^*$-node are incident to the same (two) faces.
%\myparagraph{Encoding of planar embeddings.}

\myparagraph{Encoding of planar embeddings.} A rooted SPQ$^*$-tree $T_\rho$ of an SP-graph $G$ is used to describe all planar embeddings of $G$ having the reference chain on the external face (in every planar embedding of $G$, all the edges in the chain represented by a Q$^*$-node belong to the same two faces). These embeddings are obtained by permuting in all possible ways the edges of the skeletons of the P-nodes distinct from the reference edges, around the poles.  
Namely, assume given an $st$-numbering of $G$ such that $s$ and $t$ coincide with the poles of~$\rho$.
For each P-node $\nu$ of $T_\rho$, let $u$ and $v$ be its poles where $u$ precedes $v$ in the $st$-numbering. Denote by $e_\nu$ the reference edge of $\skel(\nu)$, by $e_1, \dots, e_h$ the edges of $\skel(\nu)$ distinct from $e_\nu$, and by $\mu_1, \dots, \mu_h$ the children of~$\nu$ corresponding to $e_1, \dots, e_h$. Each permutation of $e_1, \dots, e_h$ defines a class of planar embeddings of $G_{\nu,\rho}$ with $u$ and $v$ on the external face, where the components $G_{\mu_1,\rho}, \dots, G_{\mu_h,\rho}$ are incident to $u$ and $v$ in the order of the permutation. More precisely, if $e_{i_1}, \dots, e_{i_h}$ is one of these permutations $(i_j \in \{1, \dots, h\})$, the clockwise (resp. counterclockwise) sequence of edges incident to $u$ (resp. $v$) in $\skel(\nu)$ is $e_\nu,e_{i_1}, \dots, e_{i_h}$; we say that, according to this permutation, $\mu_{i_1}, \dots, \mu_{i_h}$ and their corresponding components appear in this left-to-right order. 

\myparagraph{Independent-parallel SP-graphs.} Let $G$ be an SP-graph and let $T$ be its SPQ$^*$-tree. We say that $G$ is \emph{independent-parallel} if no two P-nodes of $T$ have a pole in common (see, e.g., Fig.~\ref{fi:preli-g}). Let $\rho$ be a Q$^*$-node of $T$. For a pole $w$ of a node $\nu$ of $T_\rho$, let $\indeg_\nu(w)$ and $\outdeg_\nu(w)$ be the degree of $w$ inside and outside~$G_{\nu,\rho}$, respectively. If $G$ is independent-parallel, each pole $w$ of a P-node~$\nu$ of~$T_\rho$ is such that $\outdeg_\nu(w)=1$; if $\nu$ is an S-node, either $\indeg_\nu(w)=1$ or $\outdeg_\nu(w)=1$. In all cases, $\outdeg_\nu(w)=1$ when $\indeg_\nu(w)>1$.

\myparagraph{Spirality.} Let $G$ be a degree-4 SP-graph and let $H$ be a rectilinear planar representation of~$G$. Let $T_\rho$ be a rooted SPQ$^*$-tree of~$G$, let $H_{\nu,\rho}$ be a component of $H$, and let $\{u,v\}$ be the poles of $\nu$, conventionally ordered according to an $st$-numbering of $G$, where $s$ and $t$ are the poles of $\rho$.
%For each pole $w \in \{u,v\}$, let $\indeg_\nu(w)$ and $\outdeg_\nu(w)$ be the degree of $w$ inside and outside $H_{\nu,\rho}$, respectively.
Since we deal with \pisps, $\outdeg_\nu(w)=1$ when $\indeg_\nu(w)>1$. Define the \emph{alias vertex} $w'$ of $w$ as follows: If $\indeg_\nu(w)=1$, then $w'=w$; else $w'$ is a dummy vertex that subdivides the edge incident to $w$ outside $H_{\nu,\rho}$. Let $P^{uv}$ be any simple path from $u$ to $v$ inside $H_{\nu,\rho}$ and let $u'$ (resp. of $v'$) be the alias vertex of $u$ (resp. $v$). The path $S^{u'v'}$ obtained concatenating $(u',u)$, $P^{uv}$, and $(v,v')$ is a \emph{spine} of $H_{\nu,\rho}$.
The \emph{spirality} $\sigma(H_{\nu,\rho})$ of $H_{\nu,\rho}$ in $H$ is the number of right turns minus the number of left turns along $S^{u'v'}$ while moving from $u'$ to~$v'$.
See, e.g., Figs.~\ref{fi:intro}~and~\ref{fi:spirality-relationships}.
%Notice that, by definition, the spirality of $H_{\nu,\rho}$ also depends on the angles at the poles of $H_{\nu,\rho}$, not only on the shape of $H_{\nu,\rho}$.
Di Battista et al.~\cite{DBLP:journals/siamcomp/BattistaLV98} show that the spirality of $H_{\nu,\rho}$ does not depend on the choice of~$P^{uv}$; also any component of~$H$ can be replaced by another component with the same spirality.
In Fig.~\ref{fi:preli-h}, the spiralities of $H_{\nu,\rho}$, $H_{\mu,\rho}$, and $H_{\phi,\rho}$ are 2, -2, and~0, respectively.
For brevity, we shall denote by $\sigma_\nu$ the spirality of a rectilinear representation of $G_{\nu,\rho}$. We say that \emph{$G_{\nu,\rho}$ admits spirality}~$\sigma_\nu$ or, equivalently, that \emph{$\nu$ admits spirality} $\sigma_\nu$, if there exists a rectilinear planar representation $H_{\nu,\rho}$ with spirality $\sigma_\nu$ in some rectilinear planar representation $H$ of $G$.
%We also say that $\nu$ \emph{admits} spirality $\sigma_\nu$ to mean that $G_{\nu,\rho}$ admits spirality~$\sigma_\nu$.

\begin{figure}[tb]
	\centering
	\subfigure[$\sigma_\nu=2$]{
		\includegraphics[height=0.26\columnwidth,page=2]{spirality-relationships.pdf}
		\label{fi:spirality-relationships-P3}
	}
	%	\hfil
	%	\subfigure[$\sigma_\nu=0$]{
	%		\includegraphics[height=0.24\columnwidth,page=3]{spirality-relationships.pdf}
	%		\label{fi:spirality-relationships-P2}
	%	}
	\hfill
	\subfigure[$\sigma_\nu=1$]{
		\includegraphics[height=0.26\columnwidth,page=1]{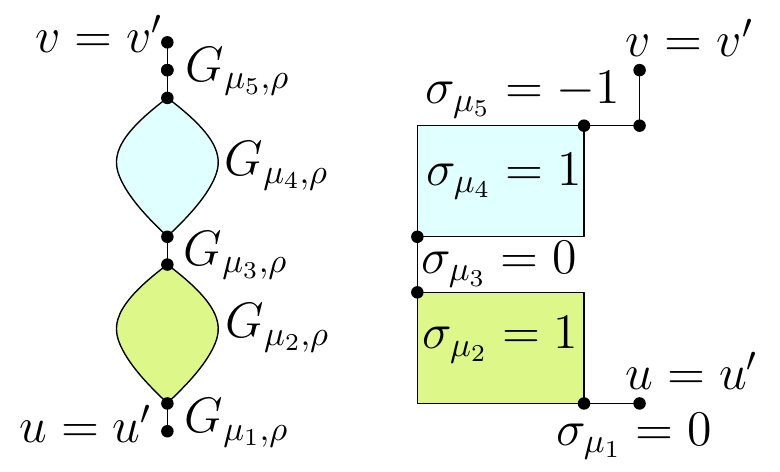}
		\label{fi:spirality-relationships-S}
	}
	\caption{Schematic illustration of the concept of spirality and of the relationships given by: (a) Lemma~\ref{le:spirality-P-node-3-children} (alias vertices are small squares); and (b) Lemma~\ref{le:spirality-S-node}.}
	\label{fi:spirality-relationships}
\end{figure}

\section{Spirality of Independent-Parallel SP-graphs}\label{se:spirality-pisp}

%The spirality of a component of a graph $G$ is a descriptor of the shapes that the component can have in a rectlinear planar representation of $G$. One of the key ingredients behind the linear-time rectlinear planarity testing for SP-graphs of vertex-degree at most three is that one can restrict the attention to a constant number of possible shapes for each component~\cite{DBLP:conf/soda/DidimoLOP20,DBLP:journals/siamdm/ZhouN08}. In this section, we first show that when we move our attention to SP-graphs with degree-4 vertices, this key ingredient no longer holds (Section~\ref{sse:lowerbound});  we then present a study of the properties of the spirality for the independent-parallel SP-graphs which makes it possible to overcome such difficulty for this family of graphs (Section~\ref{sse:intervals} and Section~\ref{se:testing}).

We first show that the components of a degree-4 SP-graph $G$ with $n$ vertices may
require a spirality that is logarithmic in $n$ (Section~\ref{sse:lowerbound}), even if $G$ is independent-parallel. 
Next, we characterize the spirality values for which the components of \pisps can be rectilinear planar (Section~\ref{sse:intervals}).
%\textcolor{blue}{Next, we prove key spirality properties for the \pisps (Section~\ref{sse:intervals}), which we exploit to test rectilinear planarity in~$O(n)$~time.}

\subsection{Spirality Lower Bound}\label{sse:lowerbound}

%\begin{restatable}{lemma}{leSpiralSecond}\label{le:spiralSecond}
%	For any constant $c>0$, there exists a rectilinear planar SP-graph~$G$ such that any rectilinear representation of~$G$ has a component with spirality larger than~$c$.
%\end{restatable}

%\begin{figure}[tb]
%	\centering
%	\subfigure[$\sigma_\nu=2$]{
%		\includegraphics[height=0.24\columnwidth,page=2]{spirality-relationships.pdf}
%		\label{fi:spirality-relationships-P3}
%	}
%%	\hfil
%%	\subfigure[$\sigma_\nu=0$]{
%%		\includegraphics[height=0.24\columnwidth,page=3]{spirality-relationships.pdf}
%%		\label{fi:spirality-relationships-P2}
%%	}
%%	\hfil
%	\subfigure[$\sigma_\nu=1$]{
%		\includegraphics[height=0.24\columnwidth,page=1]{spirality-relationships.pdf}
%		\label{fi:spirality-relationships-S}
%	}
%	\caption{Illustration for: (a) Lemma~\ref{le:spirality-P-node-3-children}; (b) Lemma~\ref{le:spirality-S-node}.}
%	\label{fi:spirality-relationships}
%\end{figure}

The proof of the lower bound (Theorem~\ref{th:LowerBound}) uses Lemma~\ref{le:spirality-P-node-3-children},
%which is proved in~\cite{DBLP:journals/siamcomp/BattistaLV98} and
which relates the spirality of a P-component with three children to the spiralities of its child components; see Fig.~\ref{fi:spirality-relationships-P3} for an illustration.

\begin{lemma}[\cite{DBLP:journals/siamcomp/BattistaLV98}]\label{le:spirality-P-node-3-children}
	Let $\nu$ be a P-node of $T_\rho$ with three children $\mu_l$, $\mu_c$, and $\mu_r$. $G_{\nu,\rho}$ admits spirality $\sigma_\nu$ with $G_{\mu_l,\rho}$, $G_{\mu_c,\rho}$, $G_{\mu_r,\rho}$ in this left-to-right order if and only if there exist three values $\sigma_{\mu_l}$, $\sigma_{\mu_c}$, and $\sigma_{\mu_r}$ such that: (i) $G_{\mu_l,\rho}$, $G_{\mu_c,\rho}$, $G_{\mu_r,\rho}$ admit spirality $\sigma_{\mu_l}$, $\sigma_{\mu_c}$, $\sigma_{\mu_r}$, respectively; and (ii) $\sigma_\nu = \sigma_{\mu_l} - 2 = \sigma_{\mu_c} = \sigma_{\mu_r} + 2$.
\end{lemma}

\begin{restatable}{theorem}{thLowerBound}\label{th:LowerBound}
	For infinitely many integer values of $n$, there exists an $n$-vertex \pisp for which every rectilinear planar representation has a component with spirality $\Omega(\log n)$.
\end{restatable}
\begin{proof}
	For any arbitrarily large even integer $N \geq 2$, we construct an \pisp $G$ with $n=O(3^N)$ vertices such that every rectilinear planar representation of $G$ has a component with spirality larger than~$N$.
	%We show that for any arbitrarily large even integer $N > 0$, there exists a rectilinear \pisp $G$ with $n=O(3^N)$ vertices such that every rectilinear planar representation of $G$ has a component with spirality larger than $N$.
	Let $L = \frac{N}{2}+1$. For any $k \in \{0, \dots, L\}$, let $G_k$ be the SP-graph inductively defined as follows: $(i)$~$G_0$ is a chain of $N+4$ vertices; $(ii)$~$G_1$ is a parallel of three copies of $G_0$, with coincident poles (Fig.~\ref{fi:LB-G1}); $(iii)$ for $k\geq2$, $G_k$ is a parallel composition of three series, each starting and ending with an edge, and having $G_{k-1}$ in the middle (Fig.~\ref{fi:LB-GK}).
	The graph $G$ is obtained by composing in a cycle two chains $p_1$ and $p_2$, of two edges each, with two copies of $G_L$ (Fig.~\ref{fi:LB-G}). The graph $G_L$ for $N=4$ is in Fig.~\ref{fi:LB-GExample}. About the number $n$ of vertices of $G$, let $n_k$ be the number of vertices of $G_k$. We have $n_0 =N+4$ and $n_k = O(3^kN)$ for $k \leq N$.
	Hence, $n_L = O(3^{\frac{N}{2}} N)$ and, since $N \leq 3^{\frac{N}{2}}$, $n_L = O(3^N)$. It follows that $n=O(3^N)$.
	%$n_k = 3^k \cdot N + 3^{k-1} \cdot 8 + 2\sum_{i=0}^{k-2} 3^i$.
	%The graph $G$ is obtained by attaching two chains $p_1$ and $p_2$ of two edges, to two copies of $G_L$, so that $p_1$ and $p_2$ alternate with the copies of $G_L$ to form a cycle; see Fig.~\ref{fi:LB-G} for a schematic illustration and Fig.~\ref{fi:LB-GExample} for an example when $N=4$.
	
	Consider first the rooted SPQ$^*$-tree $T_\rho$ of $G$, where $\rho$ represents $p_1$. All the planar embeddings of $G$ encoded by $T_\rho$ have $p_1$ (and $p_2$) on the external face of $G$, and by symmetry of the construction they are all equivalent. Any rectilinear planar representation $H$ of $G$ with an embedding encoded by $T_\rho$ requires that the restriction of~$H$ to each copy of $G_L$ has spirality zero and, at the same time, the restriction of~$H$ to one of the copies of $G_0$ in $G_L$ has spirality $N+2$. Indeed, due to Lemma~\ref{le:spirality-P-node-3-children}, for each rectilinear planar representation $H_{k}$ of~$G_{k}$, the leftmost (resp. rightmost) child component of~$H_{k}$ has spirality that is two units larger (resp. smaller) than the spirality of~$H_{k}$. Hence, if there existed a rectilinear representation of~$G_L$ with spirality greater (resp. smaller) than zero, it would contain a representation of a copy of $G_0$ with spirality greater than~$N+2$ (resp. less than $-(N+2)$), which is impossible, as the absolute value of spirality of any copy of $G_0$ is at most $N+2$. See  Fig.~\ref{fi:LB-HExample}, where $N=4$.
	
	On the other hand, if we consider the planar embeddings encoded by $T$ when rooted at a Q$^*$-node whose chain $p$ belongs to a copy of $G_L$, the same argument as above applies to the copy of $G_L$ that does not contain $p$; namely, any rectilinear representation of this copy must contain a component with spirality $N+2$.
\end{proof}

\begin{figure}[tb]
	\centering
	\subfigure[]{
		\includegraphics[height=0.23\columnwidth,page=1]{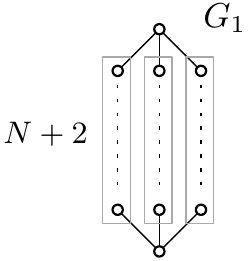}
		\label{fi:LB-G1}
	}
	\hfill
	\subfigure[]{
		\includegraphics[height=0.23\columnwidth,page=3]{LB.pdf}
		\label{fi:LB-GK}
	}
	\hfill
	\subfigure[]{
		\includegraphics[height=0.23\columnwidth,page=4]{LB.pdf}
		\label{fi:LB-G}
	}
	\subfigure[]{
		\includegraphics[height=0.29\columnwidth,page=2]{LB.pdf}
		\label{fi:LB-GExample}
	}
	\hfill
	\subfigure[]{
		\includegraphics[height=0.29\columnwidth,page=2]{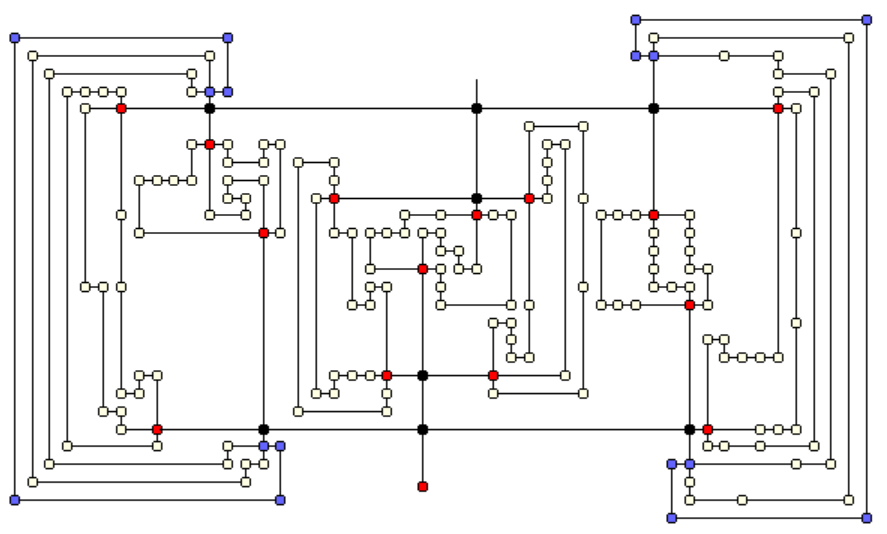}
		\label{fi:LB-HExample}
	}
	\caption{(a)--(c) The graph family of Theorem~\ref{th:LowerBound}; (d) $G_L=G_3$ ($N=4$, $L=\frac{N}{2}+1$); (e) A rectilinear planar representation of $G_L$ (automatically computed by the GDToolkiy library~\cite{DBLP:reference/crc/BattistaD13}); the two $G_0$ components with blue vertices have spirality $N+2=6$ (left) and $-(N+2)=-6$ (right), respectively.}
	\label{fi:LB}
\end{figure}

%\begin{figure}[tb]
%	\centering
%	\subfigure[]{
%		\includegraphics[height=0.29\columnwidth,page=2]{LB.pdf}
%		\label{fi:LB-GExample}
%	}
%	\hfill
%	\subfigure[]{
%		\includegraphics[height=0.29\columnwidth,page=2]{LB-HExample.png}
%		\label{fi:LB-HExample}
%	}
%	\caption{(a) The graph $G_L=G_3$ ($N=4$, $L=\frac{N}{2}+1$); (b) A rectilinear planar representation of $G_L$; the two $G_0$ components with blue vertices have spirality $N+2=6$ (left) and $-(N+2)=-6$ (right), respectively. The representation has been automatically computed by using the GDToolkiy library~\cite{DBLP:reference/crc/BattistaD13}.}
%	\label{fi:LB-Example}
%\end{figure}

%By Lemma~\ref{le:spiralFirst} and Theorem~\ref{th:LowerBound}, one could think that a linear-time rectilinear planarity testing algorithm based on a bottom-up visit of the SPQ$^*$-tree of an SP-graph may not exist.
%However, we prove in the next sections that a deeper study of the properties of the rectilinear spiraliy sets makes it possible to design a linear-time rectilinear planarity ing for the \pisps. It may be worth remarking that the SP-graphs in Lemma~\ref{le:spiralFirst} and Theorem~\ref{th:LowerBound} are~independent-parallel.
%In the following, an independent-parallel SP-graph will be simply called a \emph{\pisp}.

%%%% INTERVAL
\subsection{Rectilinear Spirality Sets}\label{sse:intervals}
%% definition of rectilinear spirality set
%% ---------------------------------------
Let $G$ be a rectilinear planar SP-graph, $T_\rho$ be a rooted SPQ$^*$-tree of $G$, and  $\nu \neq \rho$ be a node of $T_\rho$. The \emph{rectilinear spirality set} $\Sigma_{\nu,\rho}$ of $\nu$ in $T_\rho$ (and of~$G_{\nu,\rho}$) is the set of spirality values for which $G_{\nu,\rho}$ admits a rectilinear planar representation.
We will prove that there is some regularity in the rectilinear spirality sets of \pisps. %As we are going to show in Section~\ref{se:final}, such a regularity does not hold for general SP-graphs.
Denote by $\Sigma^+_{\nu,\rho}$ (resp. $\Sigma^-_{\nu,\rho}$) the subset of non-negative (resp. non-positive) values of $\Sigma_{\nu,\rho}$. Clearly, $\Sigma_{\nu,\rho} = \Sigma^+_{\nu,\rho} \cup \Sigma^-_{\nu,\rho}$.
%and $\Sigma^-_{\nu,\rho}$ the \emph{non-positive rectilinear spirality set} of~$\nu$ in $T_\rho$ (or of $G_{\nu,\rho}$).
Note that, for any value $\sigma_\nu \in \Sigma_{\nu,\rho}$, we also have that $-\sigma_\nu \in \Sigma_\nu$. Indeed, if $G_{\nu,\rho}$ admits a rectilinear representation with spirality $\sigma_\nu$ for some embedding, by flipping this embedding around the poles of $G_{\nu,\rho}$, we can obtain a rectilinear representation of $G_{\nu,\rho}$ with spirality $-\sigma_\nu$.
Hence, $\sigma_\nu \in \Sigma^+_{\nu,\rho}$ if and only if $-\sigma_\nu \in \Sigma^-_{\nu,\rho}$, and we can restrict the study of the properties of~$\Sigma_{\nu,\rho}$~to~$\Sigma^+_{\nu,\rho}$, which we call the \emph{non-negative rectilinear spirality set} of~$\nu$ in $T_\rho$ (or of $G_{\nu,\rho}$).

%\subsection{Structures of Rectilinear Spirality Sets}\label{sse:set-structures}
The main result of this subsection is Theorem~\ref{th:spirality-sets}, which proves that if $G$ is an \pisp, there is a limited number of possible structures for the sets $\Sigma^+_{\nu,\rho}$.
%To define these structures, we introduce the following notation. 
Let $m$ and $M$ be two non-negative integers such that $m < M$:
%\begin{itemize}
%	$(i)$ $[M]$ denotes the singleton $\{M\}$, and we call it a \emph{trivial interval}.
%	$(ii)$ $[m,M]^1$ denotes the set of all integers in the interval $[m,M]$. i.e., $\{m, m+1, \dots, M-1, M\}$; we call $[m,M]^1$ a \emph{jump-1 interval}.
%	$(iii)$ If $m$ and $M$ have the same parity, $[m,M]^2$ denotes the set of values $\{m, m+2, \dots, M-2, M\}$; we call $[m,M]^2$ a \emph{jump-2 interval}.
%\end{itemize}
$(i)$ $[M]$ is a \emph{trivial interval} and denotes the singleton $\{M\}$;
$(ii)$ $[m,M]^1$ is a \emph{jump-1 interval} and denotes the set of all integers in the interval $[m,M]$. i.e., $\{m, m+1, \dots, M-1, M\}$;
$(iii)$ If $m$ and $M$ have the same parity, $[m,M]^2$ is a \emph{jump-2 interval} and denotes the set of values $\{m, m+2, \dots, M-2, M\}$.

\begin{theorem}\label{th:spirality-sets}
	Let $G$ be a rectilinear planar \pisp and let $G_{\nu,\rho}$ be a component of $G$. The non-negative rectilinear spirality set $\Sigma^+_{\nu,\rho}$ of~$G_{\nu,\rho}$ has one the following six structures: $[0]$, $[1]$, $[1,2]^1$, $[0,M]^1$, $[0,M]^2$, $[1,M]^2$.
\end{theorem}

\noindent Theorem~\ref{th:spirality-sets} relies on some key technical results.
%namely Lemma~\ref{le:intervalSupportFirst}, Corollary~\ref{co:intervalSupport}, and Lemma~\ref{le:intervalSupportSecond}.
%
%We say that a component $G_{\nu,\rho}$ \emph{admits} spirality~$\sigma_\nu$ if $\sigma_\nu \in \Sigma_{\nu,\rho}$.
%For brevity we sometimes say that $\nu$ \emph{admits} spirality $\sigma_\nu$ to mean that $G_{\nu,\rho}$ admits spirality~$\sigma_\nu$.
%
Lemma~\ref{le:spirality-S-node} relates the spirality of an S-component to those of its child components (see~Fig.~\ref{fi:spirality-relationships-S}).

\begin{lemma}[\cite{DBLP:journals/siamcomp/BattistaLV98}]\label{le:spirality-S-node}
	Let $\nu$ be an S-node of $T_\rho$ with children $\mu_1, \dots, \mu_h$. The component $G_{\nu,\rho}$ admits spirality $\sigma_\nu$ if and only if $\sigma_\nu = \sum_{i=1}^{h}\sigma_{\mu_i}$, where $\sigma_{\mu_i}$ is a spirality value admitted by $G_{\mu_i,\rho}$ $(1 \leq i \leq h)$.
\end{lemma}

\noindent The next lemmas refer to (components of)~\pisps. As for Lemmas~\ref{le:spirality-P-node-3-children} and~\ref{le:spirality-S-node}, the next lemma shows the relationship between the spirality of a parallel component with two children and the spiralities of its child components. In the lemma, if $H_{\nu,\rho}$ is a rectilinear representation of $G_{\nu,\rho}$ with spirality $\sigma_\nu$, the value $\alpha_w^l$ (resp. $\alpha_w^r$) is used to represent the left (resp. right) outside angle at $w$; namely, $\alpha_w^l = 0$ (resp. $\alpha_w^r = 0$) if the left (resp. right) angle at $w$ is of $180^\circ$. Conversely, $\alpha_w^l = 1$ (resp. $\alpha_w^r = 1$) if the left (resp. right) angle at $w$ is of $90^\circ$.

\begin{lemma}[\cite{DBLP:journals/siamcomp/BattistaLV98}]\label{le:spirality-P-node-2-children}
	Let $\nu$ be a P-node of $T_\rho$ with two children $\mu_l$ and $\mu_r$, and with poles $u$ and $v$. The component $G_{\nu,\rho}$ admits spirality $\sigma_\nu$ with $G_{\mu_l,\rho}$ and $G_{\mu_r,\rho}$ in this left-to-right order if and only if there exist six values $\sigma_{\mu_l}$, $\sigma_{\mu_r}$, $\alpha_u^l$, $\alpha_u^r$, $\alpha_v^l$, and $\alpha_v^r$ such that: (i) $G_{\mu_l,\rho}$ and $G_{\mu_r,\rho}$ admit spirality $\sigma_{\mu_l}$ and $\sigma_{\mu_r}$, respectively; (ii) $\alpha_w^l \in \{0,1\}$, $\alpha_w^r \in \{0,1\}$, and $1 \leq \alpha_w^l+\alpha_w^r \leq 2$ for any $w \in \{u,v\}$; and (iii) $\sigma_\nu = \sigma_{\mu_l} - \alpha_{u}^l -  \alpha_{v}^l = \sigma_{\mu_r} + \alpha_{u}^r +\alpha_{v}^r$.
\end{lemma}	

\begin{figure}[!h]
	\centering
	\includegraphics[height=0.3\columnwidth,page=3]{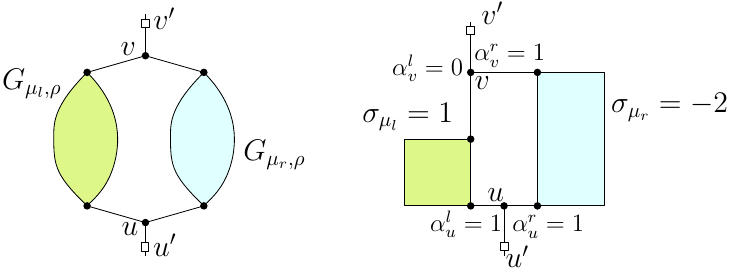}
	\caption{Illustration for Lemma~\ref{le:spirality-P-node-2-children}: the component has spirality 0; alias vertices are small squares.}
	\label{fi:spirality-relationships-P2}
\end{figure}

\begin{restatable}{lemma}{leIntervalSupportFirst}\label{le:intervalSupportFirst}
	Let $G_{\nu,\rho}$ be a component that admits spirality $\sigma_\nu \geq 2$. The following properties hold: (a) if $\sigma_\nu=2$, $G_{\nu,\rho}$ admits spirality $\sigma'_\nu=0$ or $\sigma'_\nu=1$; (b) if $\sigma_\nu>2$, $G_{\nu,\rho}$ admits spirality $\sigma'_\nu = \sigma_\nu-2$; (c) if $\sigma_\nu = 4$, $G_{\nu,\rho}$ admits spirality~$\sigma'_\nu = 0$.
\end{restatable}
\begin{proof}
	The proof is by induction on the depth of the subtree of $T_\rho$ rooted at $\nu$. In the base case $\nu$ is a Q$^*$-node and the three properties trivially hold for~$G_{\nu,\rho}$. In the inductive case, $\nu$ is either an S-node or a P-node. 
	
	%In this sketch we prove Properties~(a)--(c) for P-nodes with three children and Property~(a) for S-nodes. The remaining cases are proved similarly (App.~\ref{app:spirality-pisp}).
	
%	Suppose that the three properties hold for the children of $G_{\nu,\rho}$.
%	Note that $G_{\nu,\rho}$ can be either a S-component or a P-component with two or three children.

		\smallskip\noindent
	\textsf{-- $\nu$ is an S-node.} We inductively prove the three properties.
	
	\smallskip\noindent
	\textsf{Proof of Property (a).} If $\nu$ admits spirality $\sigma_\nu=2$, by Lemma~\ref{le:spirality-S-node}, $\nu$ has a child $\mu$ that admits spirality $\sigma_\mu>0$. If $\sigma_\mu=1$, $\mu$ also admits spirality -1, and $\nu$ admits spirality 0. If $\sigma_\mu=2$, by inductively using Property~(a), $\mu$ also admits 0 or 1, and so does $\nu$. If $\sigma_\mu>2$, by inductively using Property~(b), $\mu$ admits spirality $\sigma_\mu-2$, and $\nu$ admits spirality~0.     
	
	\smallskip\noindent
	\textsf{Proof of Property (b).} If $\nu$ admits spirality $\sigma_\nu >2$, by Lemma~\ref{le:spirality-S-node}, $\nu$ has child~$\mu$ that admits spirality $\sigma_\mu>2$. By inductively using Property~(b), $\mu$ admits spirality $\sigma_\mu-2$ and by Lemma~\ref{le:spirality-S-node} $\nu$ admits spirality $\sigma_\nu-2$. If $\nu$ has child $\mu$ such that $\mu$ admits spirality 1, then $\mu$ admits spirality -1, and $\nu$ admits spirality $\sigma_\nu - 2$. Else, $\nu$ has two children $\mu_1$ and $\mu_2$ such that $\mu_1$ and $\mu_2$ both admit spirality 2. By inductively using Property~(a), either one of them also admits spirality 0 or they both admit spirality 1. In any case, $\nu$ admits spirality $\sigma_\nu - 2$.
	
	\smallskip\noindent
	\textsf{Proof of Property (c).} If $\nu$ admits spirality $\sigma_\nu = 4$, by Lemma~\ref{le:spirality-S-node} one of the following cases applies: (i) $\nu$ has a child $\mu$ that admits spirality 4; if so, by inductively using Property~(c), $\nu$ admits spirality 0. (ii) $\nu$ has a child $\mu$ that admits spirality $\sigma_\mu > 4$; if so, by inductively applying Property~(b) twice, $\mu$ admits spirality $\sigma_\mu - 4$, and hence $\nu$ admits spirality 0. 
	(iii) $\nu$ has two children $\mu_1$ and $\mu_2$ such that each of them admits spirality either 1 or 3; observe that if $\mu_i$ $(i \in \{1,2\})$ admits spirality 1, it also admits spirality -1 and if $\mu_i$ admits spirality 3, it also admits spirality 1 by inductively using Property~(b); this implies that $\nu$ admits spirality $\sigma_\nu-4$. 
	
		\smallskip \noindent
	\textsf{-- $\nu$ is a P-node with three children.}
	Let $H_{\nu,\rho}$ be a rectilinear planar representation of $G_{\nu,\rho}$ with spirality $\sigma_\nu$.  Let $\mu_l$, $\mu_c$, and $\mu_r$ be the children of $\nu$ such that $G_{\mu_l,\rho}$, $G_{\mu_c,\rho}$, and  $G_{\mu_r,\rho}$ appear in this left-to-right order in $H_{\nu,\rho}$. By Lemma~\ref{le:spirality-P-node-3-children}, $\sigma_{\mu_l}=\sigma_\nu+2$, $\sigma_{\mu_c}=\sigma_\nu$, and $\sigma_{\mu_r}=\sigma_\nu-2$.
	
	\smallskip \noindent
	\textsf{Proof of Property (a).} If $\sigma_\nu=2$, we have $\sigma_{\mu_l}=4$, $\sigma_{\mu_c}=2$, and $\sigma_{\mu_r}=0$; see Fig.~\ref{fig:3children2-0-a}. By inductively using Property~(b), $\mu_l$ admits spirality~2. Also $\mu_c$ admits spirality~$-2$. Hence, exchanging $G_{\mu_c}$ with $G_{\mu_r}$ in the left-to-right order, by Lemma~\ref{le:spirality-P-node-3-children}, $\nu$ admits spirality 0; see~Fig.~\ref{fig:3children2-0-d}.
	
	\smallskip\noindent
	\textsf{Proof of Property (b).} If $\sigma_\nu>2$, we distinguish three cases:
	
	\noindent (i) $\sigma_\nu=3$, which implies $\sigma_{\mu_l}=5$, $\sigma_{\mu_c}=3$, and $\sigma_{\mu_r}=1$; see Fig.~\ref{fig:3children2-0-b}. By inductively using Property~(b), $\mu_l$ and $\mu_c$ admit spirality 3 and 1, respectively. Also, $\mu_r$ admits spirality -1. By Lemma~\ref{le:spirality-P-node-3-children}, $\nu$ admits spirality $\sigma_\nu-2=1$; see~Fig.~\ref{fig:3children2-0-e}.
	
	\noindent (ii) $\sigma_\nu=4$, which implies $\sigma_{\mu_l}=6$, $\sigma_{\mu_c}=4$, and $\sigma_{\mu_r}=2$see Fig.~\ref{fig:3children2-0-c}. By inductively using Property~(b), $\mu_l$ admits spirality 4; also, by inductively using Property~(c), $\mu_c$ admits spirality 0. Hence, exchanging $G_{\mu_c}$ with $G_{\mu_r}$ in the left-to-right order, by Lemma~\ref{le:spirality-P-node-3-children}, $\nu$ admits spirality $\sigma_\nu-2=2$; see~Fig.~\ref{fig:3children2-0-f}.
	
	\begin{figure}[t]
		\centering
		\subfigure[$\sigma_\nu=2$]{\includegraphics[width=0.26\columnwidth,page=1]{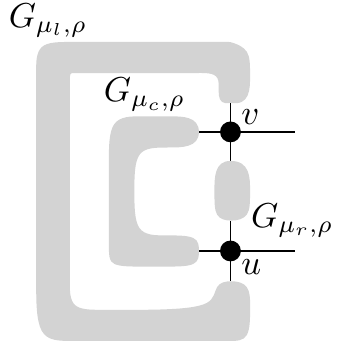}\label{fig:3children2-0-a}}
		\hfil
		\subfigure[$\sigma_\nu=3$]{\includegraphics[width=0.26\columnwidth,page=15]{nojumping3.pdf}\label{fig:3children2-0-b}}
		\hfil
		\subfigure[$\sigma_\nu=4$]{\includegraphics[width=0.26\columnwidth,page=14]{nojumping3.pdf}\label{fig:3children2-0-c}}
		\hfil
		\subfigure[$\sigma'_\nu=0$]{\includegraphics[width=0.26\columnwidth,page=2]{nojumping3.pdf}\label{fig:3children2-0-d}}
		\hfil
		\subfigure[$\sigma'_\nu=1$]{\includegraphics[width=0.26\columnwidth,page=16]{nojumping3.pdf}\label{fig:3children2-0-e}}
		\hfil
		\subfigure[$\sigma'_\nu=2$]{\includegraphics[width=0.26\columnwidth,page=17]{nojumping3.pdf}\label{fig:3children2-0-f}}
		\hfil
		\caption{Illustration of Lemma~\ref{le:intervalSupportFirst} for a P-node with three children.}
		\label{fig:3children2-0}
	\end{figure}
	
	\noindent (iii) $\sigma_\nu>4$, which implies $\sigma_{\mu_r}>2$. By inductively using Property~(b), $\mu_l$, $\mu_c$, $\mu_r$ admit spirality $\sigma_{\mu_l}-2$, $\sigma_{\mu_c}-2$, $\sigma_{\mu_r}-2$, and hence $\nu$ admits spirality~$\sigma_\nu-2$.
	
	\smallskip
	\noindent \textsf{Proof of Property (c).} If $\sigma_\nu = 4$ then
	$\sigma_{\mu_l}=6$, $\sigma_{\mu_c}=4$, $\sigma_{\mu_r}=2$.
	By inductively using Property~(b) twice, $\mu_l$ admits spirality 2. By inductively using  Property~(c), $\mu_c$ admits spirality 0. Since $\mu_r$ admits spirality -2, %by Lemma~\ref{le:spirality-P-node-3-children},
	$\nu$ admits spirality $\sigma_\nu-4=0$.

	\smallskip\noindent
	\textsf{-- $\nu$ is a P-node with two children}
	
	\smallskip\noindent
	Let $H_{\nu,\rho}$ be a rectilinear planar representation of $G_{\nu,\rho}$ with spirality $\sigma_\nu$. Let  $G_{\mu_l,\rho}$ and  $G_{\mu_r,\rho}$ be the left child and the right child of $G_{\nu,\rho}$ in $H_{\nu,\rho}$, respectively. By Lemma~\ref{le:spirality-P-node-2-children} we have $\sigma_\nu = \sigma_{\mu_l} - \alpha_{u}^l -  \alpha_{v}^l = \sigma_{\mu_r} +  \alpha_{u}^r +\alpha_{v}^r$. Lemma~\ref{le:spirality-P-node-2-children} implies $\sigma_{\mu_l}-\sigma_{\mu_r}\in [2,4]$. Without loss of generality, we assume that $\alpha_v^l\ge \alpha_u^l$.
	
	\smallskip\noindent
	\textsf{Proof of Property (a).} If $\sigma_\nu=2$, we distinguish three cases depending on the value of $\sigma_{\mu_l}-\sigma_{\mu_r}$. Suppose first that $\sigma_{\mu_l}-\sigma_{\mu_r}=2$. There are three subcases:
	\begin{itemize}
		\item[(i)] $\sigma_{\mu_l}=2$, $\sigma_{\mu_r}=0$, and $\alpha_u^l=\alpha_v^l=0$; see Fig.~\ref{fig:2children2-0-a}. For $\alpha_u^l=\alpha_v^l=1$ and $\alpha_u^r=\alpha_v^r=0$, by Lemma~\ref{le:spirality-P-node-2-children}, $G_{\nu,\rho}$ admits spirality $\sigma_\nu-2=0$; see Fig.~\ref{fig:2children2-0-b}.
		\item[(ii)] $\sigma_{\mu_l}=3$, $\sigma_{\mu_r}=1$, $\alpha_v^r=0$, and $\alpha_u^l=0$; see Fig.~\ref{fig:2children2-0-c}. By inductively using Property~(b), $G_{\mu_l,\rho}$ admits spirality 1. Also, $G_{\mu_r,\rho}$ admit spirality $\sigma_{\mu_r}=-1$. For $\alpha_u^l=\alpha_v^r=0$ (which implies $\alpha_u^r=\alpha_v^l=1$), by Lemma~\ref{le:spirality-P-node-2-children}, $G_{\nu,\rho}$ admits spirality 0; see Fig.~\ref{fig:2children2-0-d}.
		\item[(iii)] $\sigma_{\mu_l}=4$, $\sigma_{\mu_r}=2$, and $\alpha_u^r=\alpha_v^r=1$; see Fig.~\ref{fig:2children2-0-e}. By inductively using Property~(c), $G_{\mu_l,\rho}$ admits spirality 0. Hence, exchanging $G_{\mu_l,\rho}$ and $G_{\mu_r,\rho}$ in the left-to-right order, and for  
		$\alpha_u^r=\alpha_v^r=0$, by Lemma~\ref{le:spirality-P-node-2-children}, $G_{\nu,\rho}$ admits spirality 0; see Fig.~\ref{fig:2children2-0-f}.
	\end{itemize}
	
	\begin{figure}
		\centering
		\subfigure[$\sigma_\nu=2$]{\includegraphics[width=0.27\columnwidth,page=3]{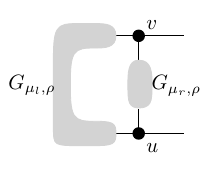}
			\label{fig:2children2-0-a}
		}
		\hfil
		\subfigure[$\sigma'_\nu=0$]{\includegraphics[width=0.27\columnwidth,page=4]{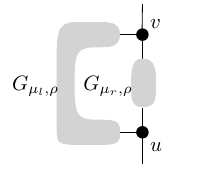} 
			\label{fig:2children2-0-b}
		}
		\hfil
		\subfigure[$\sigma_\nu=2$]{\includegraphics[width=0.27\columnwidth,page=5]{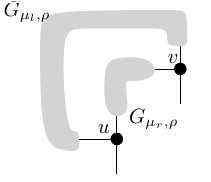}
			\label{fig:2children2-0-c}
		}
		\hfil
		\subfigure[$\sigma'_\nu=0$]{\includegraphics[width=0.27\columnwidth,page=6]{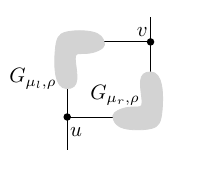}
			\label{fig:2children2-0-d}
		}
		\hfil
		\subfigure[$\sigma_\nu=2$]{\includegraphics[width=0.27\columnwidth,page=7]{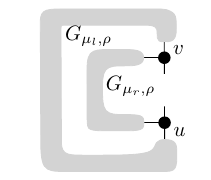}
			\label{fig:2children2-0-e}
		}
		\hfil
		\subfigure[$\sigma'_\nu=0$]{\includegraphics[width=0.27\columnwidth,page=8]{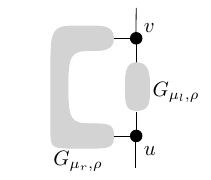}
			\label{fig:2children2-0-f}
		}
		\hfil	
		\caption{Illustration for the proof Property~(a) of Lemma~\ref{le:intervalSupportFirst} for a P-component with two children for the case $\sigma_{\mu_l}-\sigma_{\mu_r}=2$.}
		\label{fig:2children2-0}
	\end{figure}
	\noindent
	
	Suppose now that $\sigma_{\mu_l}-\sigma_{\mu_r}=3$. In this case, for one of the two poles $\{u,v\}$ of $\nu$, say $v$, we have $\alpha_v^l=\alpha_v^r=1$. There are two subcases:
	\begin{itemize}
		\item[(iv)] $\sigma_{\mu_l}=3$ and $\sigma_{\mu_r}=0$; see Fig.~\ref{fig:2childrenDiffspir34-a}. In this case $\alpha_u^r=1$.
		For $\alpha_u^r=0$, by Lemma~\ref{le:spirality-P-node-2-children}, $G_{\nu,\rho}$ admits spirality 1; see Fig.~\ref{fig:2childrenDiffspir34-b}.
		\item[(v)] $\sigma_{\mu_l}=4$ and $\sigma_{\mu_r}=1$; see Fig.~\ref{fig:2childrenDiffspir34-c}. By inductively using Property~(b), $G_{\mu_l,\rho}$ admits spirality 2. Also, $G_{\mu_r,\rho}$ admits spirality -1. For $\alpha_u^r=0$, by Lemma~\ref{le:spirality-P-node-2-children}, $G_{\nu,\rho}$ admits spirality 0; see Fig.~\ref{fig:2childrenDiffspir34-d}.
	\end{itemize}
	\begin{figure}
		\centering
		\subfigure[$\sigma_\nu=2$]{\includegraphics[width=0.27\columnwidth,page=9]{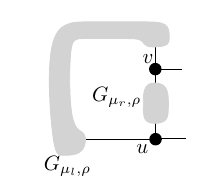}
			\label{fig:2childrenDiffspir34-a}
		}
		\hfil
		\subfigure[$\sigma'_\nu=1$]{\includegraphics[width=0.27\columnwidth,page=10]{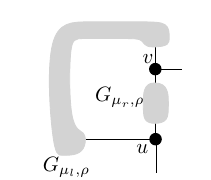}
			\label{fig:2childrenDiffspir34-b}
		}
		\hfil
		\subfigure[$\sigma_\nu=2$]{\includegraphics[width=0.27\columnwidth,page=11]{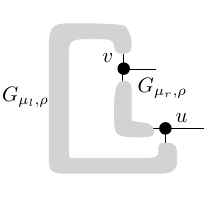}
			\label{fig:2childrenDiffspir34-c}
		}
		\hfil
		\subfigure[$\sigma'_\nu=0$]{\includegraphics[width=0.27\columnwidth,page=12]{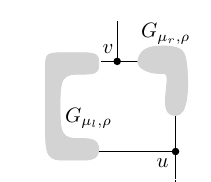}
			\label{fig:2childrenDiffspir34-d}
		}
		\hfil
		\subfigure[$\sigma_\nu=2$]{\includegraphics[width=0.27\columnwidth,page=13]{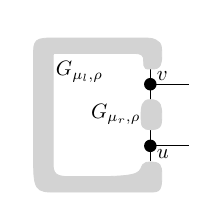}
			\label{fig:2childrenDiffspir34-e}
		}
		\hfil
		\subfigure[$\sigma'_\nu=0$]{\includegraphics[width=0.27\columnwidth,page=4]{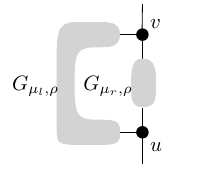}
			\label{fig:2childrenDiffspir34-f}
		}
		\hfil
		\caption{Illustration for the proof Property~(a) of Lemma~\ref{le:intervalSupportFirst} for a P-component with two children for the cases $\sigma_{\mu_l}-\sigma_{\mu_r}=3$ and $\sigma_{\mu_l}-\sigma_{\mu_r}=4$.}
		\label{fig:2childrenDiffspir34}
	\end{figure}
	
	\noindent
	Finally, suppose $\sigma_{\mu_l}-\sigma_{\mu_r}=4$; see Fig.~\ref{fig:2childrenDiffspir34-e}. We have $\sigma_{\mu_l}=4$ and $\sigma_{\mu_r}=0$. 
	By inductively using Property~(b), $G_{\mu_l,\rho}$ admits spirality 2. For $\alpha_u^l=\alpha_v^l=0$, by Lemma~\ref{le:spirality-P-node-2-children}, $G_{\nu,\rho}$
	admits spirality 0; see Figure~\ref{fig:2childrenDiffspir34}(f).

	\smallskip\noindent
	\textsf{Proof of Property (b).}  If $\sigma_\nu>2$, we have three cases: $\sigma_\nu=3$; $\sigma_\nu=4$; $\sigma_\nu>4$.  
	
	\smallskip\noindent
	{-- $\sigma_\nu=3$}. As before, we perform a case analysis based on the value of $\sigma_{\mu_l}-\sigma_{\mu_r}$. Suppose first that $\sigma_{\mu_l}-\sigma_{\mu_r}=2$. There are three subcases:
	\begin{itemize}
		\item[(i)] If $\sigma_{\mu_l}=3$ and $\sigma_{\mu_r}=1$, we have $\alpha_u^r=\alpha_v^r=0$. For $\alpha_u^r=\alpha_v^r=1$ and  $\alpha_u^l=\alpha_v^l=0$, by Lemma~\ref{le:spirality-P-node-2-children}, $G_{\nu,\rho}$ admits spirality $\sigma_\nu - 2 = 1$.
		\item[(ii)] If $\sigma_{\mu_l}=4$ and $\sigma_{\mu_r}=2$, by inductively using Property~(c), $G_{\mu_l,\rho}$ admits spirality 0. Exchanging $G_{\mu_l,\rho}$ and $G_{\mu_r,\rho}$ in the left-to-right order and for $\alpha_u^l=\alpha_v^r=0$, by Lemma~\ref{le:spirality-P-node-2-children}, $G_{\nu,\rho}$ admits spirality $\sigma_\nu - 2 = 1$.  
		\item[(iii)] If $\sigma_{\mu_l}=5$ and $\sigma_{\mu_r}=3$, by inductively using Property~(b), $G_{\mu_l,\rho}$ and $G_{\mu_r,\rho}$ admit spiralities $\sigma_{\mu_l}-2$ and $\sigma_{\mu_r}-2$, respectively. Hence $G_{\nu,\rho}$ admits spirality $\sigma_\nu - 2 = 1$.
	\end{itemize}

	Suppose now that $\sigma_{\mu_l}-\sigma_{\mu_r}=3$. As in proof for Property~(a), assume, without loss of generality, that $\alpha_v^l=\alpha_v^r=1$. The following subcases hold:
	\begin{itemize}
		\item[(iv)] If $\sigma_{\mu_l}=4$ and $\sigma_{\mu_r}=1$, by inductively using Property~(b), $G_{\mu_l,\rho}$ admits spirality 2. Also, $G_{\mu_r}$ admits spirality -1. For $\alpha_u^l=0$, by Lemma~\ref{le:spirality-P-node-2-children}, $G_{\nu,\rho}$ admits spirality $\sigma_\nu- 2 = 1$.
		\item[(v)] If $\sigma_{\mu_l}=5$ and $\sigma_{\mu_r}=2$, by inductively using Property~(a), $G_{\mu_r}$ admits spirality either 1 or 0. Suppose first that $G_{\mu_r}$ admits spirality 1. By inductively using Property~(b), $G_{\mu_l,\rho}$ admits spirality 3. For $\alpha_v^r=\alpha_u^r=0$, by Lemma~\ref{le:spirality-P-node-2-children}, $G_{\nu,\rho}$ admits spirality $\sigma_\nu- 2 = 1$. Suppose now that $G_{\mu_r,\rho}$ admits spirality 0. As before, $G_{\mu_l,\rho}$ admits spirality 3. For $\alpha_u^r=0$, we have again that $G_{\nu,\rho}$ admits spirality $\sigma_\nu- 2 = 1$; see Fig.~\ref{fig:2childrenDiffspir34}(b). 
	\end{itemize}
	
	Suppose finally that $\sigma_{\mu_l}-\sigma_{\mu_r}=4$. We have $\sigma_{\mu_l}=5$ and $\sigma_{\mu_r}=1$. By inductively using Property~(b), $G_{\mu_l,\rho}$ admits spirality 3, and then for $\alpha_v^r=0$ and $\alpha_u^r=0$, we have that $G_{\nu,\rho}$ admits spirality $\sigma_\nu- 2 = 1$.
	
	\smallskip\noindent
	{-- $\sigma_\nu=4$}. Consider the subcase where $\sigma_{\mu_l}-\sigma_{\mu_r}=2$. 
	\begin{itemize}
		\item[(vi)] If $\sigma_{\mu_l}=4$ and $\sigma_{\mu_r}=2$, we have $\alpha_u^l=\alpha_v^l=0$. For $\alpha_u^l=\alpha_v^l=1$ and $\alpha_u^r=\alpha_v^r=0$, by  Lemma~\ref{le:spirality-P-node-2-children}, $G_{\nu,\rho}$ admits spirality $\sigma_\nu - 2 =2$.
		\item[(vii)] If $\sigma_{\mu_l}=5$ and $\sigma_{\mu_r}=3$ or $\sigma_{\mu_l}=6$ and $\sigma_{\mu_r}=4$, by inductively using Property~(b), $G_{\mu_l,\rho}$ and $G_{\mu_r,\rho}$ admit spiralities $\sigma_{\mu_l} - 2$ and $\sigma_{\mu_r} - 2$, respectively. Hence, $G_{\nu,\rho}$ admits spirality $\sigma_\nu - 2 = 2$.    
	\end{itemize}
	
	Suppose now that $\sigma_{\mu_l}-\sigma_{\mu_r}=3$. 
	\begin{itemize}
		\item[(viii)] If $\sigma_{\mu_l}=5$ and $\sigma_{\mu_r}=2$, by inductively using Property~(b), $G_{\mu_l}$ admits spirality 3. Also, by inductively using Property~(a), $G_{\mu_r}$ admits spirality either 0 or 1. In the first case, for $\alpha_u^l=0$, we have that $G_{\nu,\rho}$ admits spirality $\sigma_\nu - 2 = 2$; see Fig.~\ref{fig:2childrenDiffspir34-a} the property holds for $\alpha_u^r=0$; Fig.~\ref{fig:2childrenDiffspir34-c}.
		
		\item[(ix)] If $\sigma_{\mu_l}=6$ and $\sigma_{\mu_r}=3$, by inductively using Property~(b), $G_{\mu_l,\rho}$ and $G_{\mu_r,\rho}$ admit spiralities $\sigma_{\mu_l} - 2$ and $\sigma_{\mu_r} - 2$, respectively. Hence, $G_{\nu,\rho}$ admits spirality $\sigma_\nu - 2 = 2$.
	\end{itemize}
	
	Suppose finally that $\sigma_{\mu_l}-\sigma_{\mu_r}=4$; we have $\sigma_{\mu_l}=6$ and $\sigma_{\mu_r}=2$. By Property~(b), $G_{\mu_l}$ admits spirality 4, and by for $\alpha_u^l=\alpha_v^l=1$, $G_{\nu,\rho}$ admits spirality $\sigma_\nu - 2 = 2$; see Fig.~\ref{fig:2children2-0}(e).  
	
	\smallskip\noindent		
	{-- $\sigma_\nu>4$}. We always have $\sigma_{\mu_r}>2$ (and $\sigma_{\mu_l}>2$). 
	By inductively using Property~(b), $G_{\mu_l,\rho}$ and $G_{\mu_r,\rho}$ admit spiralities $\sigma_{\mu_l} - 2$ and $\sigma_{\mu_r} - 2$, respectively. Hence, $G_{\nu,\rho}$ admits spirality $\sigma_\nu - 2 = 2$.

	\smallskip\noindent
	\textsf{Proof of Property (c).} If $\sigma_\nu=4$, we still consider perform a case analysis based on the value of $\sigma_{\mu_l}-\sigma_{\mu_r}$. Suppose first that $\sigma_{\mu_l}-\sigma_{\mu_r}=2$. There are three subcases.
	\begin{itemize}
		\item[(i)] Suppose $\sigma_{\mu_l}=4$ and $\sigma_{\mu_r}=2$. By inductively using Property~(c), $G_{\mu_l}$ admits spirality 0. Exchanging $G_{\mu_l,\rho}$ and $G_{\mu_r,\rho}$, and for $\alpha_u^l=\alpha_v^l=0$, we have that $G_{\nu,\rho}$ admits spirality 0.
		
		\item[(ii)] Suppose $\sigma_{\mu_l}=5$ and $\sigma_{\mu_r}=3$. By inductively using 
		Property~(b) (applied twice), $G_{\mu_l,\rho}$ and $G_{\mu_r,\rho}$ admit spirality 1; hence,  $G_{\mu_r}$ also admits -1. For $\alpha_v^r=\alpha_u^l=0$, we have that $G_{\nu,\rho}$ admits spirality 0; see Fig.~\ref{fig:2children2-0-d}.
		
		\item[(iii)] Suppose $\sigma_{\mu_l}=6$ and $\sigma_{\mu_r}=4$. 
		By inductively using Property~(b) (applied twice), $G_{\mu_l,\rho}$ admits spirality 2, and hence, by inductively using Property~(c), it also admits spirality 0. For $\alpha_u^l=\alpha_v^l=0$, $G_{\nu,\rho}$ admits spirality 0; see Fig.~\ref{fig:2children2-0-b}.
	\end{itemize}
	
	Suppose now that $\sigma_{\mu_l}-\sigma_{\mu_r}=3$. We have the following subcases. 
	\begin{itemize}
		\item[(iv)] Suppose $\sigma_{\mu_l}=5$ and $\sigma_{\mu_r}=2$. By inductively using 
		Property~(b) (applied twice), $G_{\mu_l,\rho}$ admits spirality 1. Also, $G_{\mu_l,\rho}$ admits spirality -2. For $\alpha_v^l=0$, $G_{\nu,\rho}$ admits spirality 0.
		
		\item[(v)] Suppose $\sigma_{\mu_l}=6$ and $\sigma_{\mu_r}=3$. By inductively using 
		Property~(b) (applied twice), $G_{\mu_l,\rho}$ admits spirality 2 and $G_{\mu_r,\rho}$ admits spirality 1, and hence also spirality -1. For $\alpha_u^l=0$, we have that $G_{\nu,\rho}$ admits spirality 0.
	\end{itemize}
	
	Finally, suppose that $\sigma_{\mu_l}-\sigma_{\mu_r}=4$. We have $\sigma_{\mu_l}=6$ and $\sigma_{\mu_r}=4$. By inductively using Property~(b) (applied twice), $G_{\mu_l,\rho}$ admits spirality 2, and by inductively using Property~(c), $G_{\mu_r,\rho}$ admits spirality 0. For $\alpha_u^l=\alpha_v^l=0$, $G_{\nu,\rho}$ admits spirality 0; see Fig.~\ref{fig:2children2-0-b}.

\end{proof}

%\smallskip\noindent By Properties~(b)-(c) of Lemma~\ref{le:intervalSupportFirst}, we immediately have the following.

\noindent Lemma~\ref{le:intervalSupportFirst} immediately implies Corollary~\ref{co:intervalSupport}.

\begin{corollary}\label{co:intervalSupport}
	If $G_{\nu,\rho}$ admits spirality $\sigma_\nu > 2$, $G_{\nu,\rho}$ admits spirality for every value in $[1,\sigma_\nu]^2$, when $\sigma_\nu$ is odd, or for every value in $[0,\sigma_\nu]^2$, when $\sigma_\nu$ is even.
\end{corollary}

%The proof of Lemma~\ref{le:intervalSupportThird} is in Appendix~\ref{app:spirality-pisp}.

%\begin{restatable}{lemma}{leIntervalsSupportSecond}\label{le:intervalSupportSecond}
%	Let $\nu$ be a P-node with two children and suppose that $G_{\nu,\rho}$ admits spirality $\sigma_\nu \geq 0$. There exists a rectilinear planar representation of $G_{\nu,\rho}$ with spirality $\sigma_\nu$ such that difference of spirality between the left child component and the right child component of $G_{\nu,\rho}$ is either 2 or 3.
%\end{restatable}

%\begin{restatable}{lemma}{leIntervalsSupportSecond}\label{le:intervalSupportSecond}
%	Let $\nu$ be a P-node with two children and suppose that $G_\nu$ admits spirality $\sigma_\nu \geq 0$. There exists a rectilinear planar representation $H_\nu$ of $G_\nu$ with spirality $\sigma_\nu$ such that $\sigma_{\mu_l} - \sigma_{\mu_r} \in \{2,3\}$, where $\sigma_{\mu_l}$ and $\sigma_{\mu_r}$ denote the spiralities of the left child component and of the right child component of $G_\nu$ in $H_\nu$, respectively.
%\end{restatable}

%\begin{restatable}{lemma}{leIntervalsSupportThird}\label{le:intervalSupportThird}
% 	Let $G_{\nu,\rho}$ be a component such that $\Sigma^+_{\nu,\rho}$ is not a trivial interval and has maximum value $M > 2$. If $\Sigma^+_{\nu,\rho}$ contains an integer whose parity is different from the parity of $M$, then $\Sigma^+_{\nu,\rho}=[0,M]^1$.
%\end{restatable}

\begin{figure}[t]
	\centering
	\subfigure[]{\includegraphics[height=0.179\columnwidth,page=1]{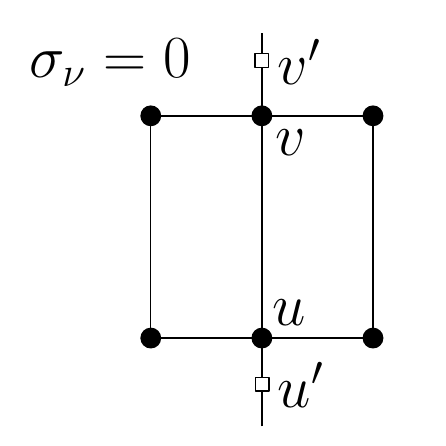}
	\label{fi:intervals-new-1}
	}
	\hfil
	\subfigure[]{\includegraphics[height=0.179\columnwidth,page=6]{intervals-new.pdf}
	\label{fi:intervals-new-6}
	}\\
	\hfil
	\subfigure[]{\includegraphics[height=0.239\columnwidth,page=2]{intervals-new.pdf}
	\label{fi:intervals-new-2}
	}
	\hfil
	\subfigure[]{\includegraphics[height=0.239\columnwidth,page=5]{intervals-new.pdf}
	\label{fi:intervals-new-5}
	}
	\hfil
	\subfigure[]{\includegraphics[height=0.239\columnwidth,page=3]{intervals-new.pdf}
	\label{fi:intervals-new-3}
	}
	\hfil
	\subfigure[]{\includegraphics[height=0.169\columnwidth,page=4]{intervals-new.pdf}
	\label{fi:intervals-new-4}
	}
	\caption{Examples of non-negative spirality sets for each of the six structures in Theorem~\ref{th:spirality-sets}: (a) $[0]$; (b) $[1]$; (c) $[1,2]^1$; (d) $[0,2]^1$; (e) $[1,3]^2$; (f) $[0,2]^2$.}
	\label{fi:intervals-new}
\end{figure}

The next lemma states an interesting property that is used to prove Lemma~\ref{le:intervalSupportThird}.

%% SECOND
%\leIntervalsSupportSecond*
\begin{restatable}{lemma}{leIntervalsSupportSecond}\label{le:intervalSupportSecond}
	Let $\nu$ be a P-node with two children and suppose that $G_{\nu,\rho}$ admits spirality $\sigma_\nu \geq 0$. There exists a rectilinear planar representation of $G_{\nu,\rho}$ with spirality $\sigma_\nu$ such that the difference of spirality between the left child component and the right child component of $G_{\nu,\rho}$ is either 2 or 3.
\end{restatable}
\begin{proof}
	Let $H_{\nu,\rho}$ be any rectilinear planar representation of $G_{\nu,\rho}$ with spirality~$\sigma_\nu$. Also, let $\sigma_{\mu_l}$ and $\sigma_{\mu_r}$ be the spiralities of the left child component $H_{\mu_l,\rho}$ and of the right child component $H_{\mu_r,\rho}$ of $H_{\nu,\rho}$, respectively. Let $G_{\mu_l,\rho}$ and $G_{\mu_r,\rho}$ be the underlying graphs of $H_{\mu_l,\rho}$ and $H_{\mu_r,\rho}$.    
	By Lemma~\ref{le:spirality-P-node-2-children}, we have $2 \leq \sigma_{\mu_l}-\sigma_{\mu_r} \leq 4$. We show that if $\sigma_{\mu_l}-\sigma_{\mu_r} = 4$, one can construct a representation $H'_{\nu,\rho}$ of $G_{\nu,\rho}$ with spirality $\sigma'_\nu=\sigma_\nu$ such $\sigma'_{\mu_l}-\sigma'_{\mu_r}\in [2,3]$. 
	Since $\sigma_{\mu_l}-\sigma_{\mu_r}=4$, we have 	${\alpha}_u^l={\alpha}_v^l={\alpha}_u^r={\alpha}_v^r=1$, where $u$ and $v$ are the poles of $\nu$. We distinguish between two cases:
	
	\smallskip\noindent{-- Case $\sigma_\nu=0$:} In this case, $\sigma_{\mu_l}=2$ and $\sigma_{\mu_r}=-2$. See Fig.~\ref{fi:spir-nodiff4-a}. By Property~(a) of Lemma~\ref{le:intervalSupportFirst}, both $G_{\mu_l,\rho}$ and $G_{\mu_r,\rho}$ admit spirality 0 or 1. Assume first that $G_{\mu_l,\rho}$ admits spirality~1. We can construct $H'_{\nu,\rho}$ by merging in parallel two representations $H'_{\mu_l,\rho}$ of $G_{\mu_l,\rho}$ and $H'_{\mu_r,\rho}$ of $G_{\mu_r,\rho}$ (in the same left-to-right order they have in $H_{\nu,\rho}$) in such a way that: $H'_{\mu_l,\rho}$ has spirality $\sigma'_{\mu_l}=1$, $\sigma'_{\mu_r}=\sigma_{\mu_r}=-2$, ${\alpha'}_u^l=0$, and ${\alpha'}_v^l={\alpha'}_u^r={\alpha'}_v^r=1$; see Figure~\ref{fi:spir-nodiff4-b}. Assume now that $G_{\mu_l,\rho}$ does not admit spirality~1 but admits spirality~0. We can construct $H'_{\nu,\rho}$ by merging in parallel two representations $H'_{\mu_l,\rho}$ of $G_{\mu_l,\rho}$ and $H'_{\mu_r,\rho}$ of $G_{\mu_r,\rho}$ (in the same left-to-right order they have in $H_{\nu,\rho}$) in such a way that: $H'_{\mu_l,\rho}$ has spirality $\sigma'_{\mu_l}=0$, $\sigma'_{\mu_r}=\sigma_{\mu_r}=-2$, ${\alpha'}_u^l={\alpha'}_v^l=0$, and  ${\alpha'}_u^r={\alpha'}_v^r=1$; see Fig.~\ref{fi:spir-nodiff4-c}. In both cases $H'_{\nu,\rho}$ has spirality $\sigma'_\nu = \sigma_\nu$ and $\sigma'_{\mu_l}-\sigma'_{\mu_r} \in [2,3]$.  
	
	\smallskip\noindent{-- Case $\sigma_\nu>0$:} In this case, $\sigma_{\mu_l} > 3$ (because $\sigma_\nu = \sigma_{\mu_l} - {\alpha}_u^l - {\alpha}_v^l$ by Lemma~\ref{le:spirality-P-node-2-children}, and ${\alpha}_u^l + {\alpha}_v^l = 2$ by hypothesis). See Fig.~\ref{fi:spir-nodiff4-d}, where $\sigma_\nu=2$. Hence, by Property~(b) of Lemma~\ref{le:intervalSupportFirst}, $G_{\mu_l,\rho}$ admits spirality $\sigma'_{\mu_l}=\sigma_{\mu_l}-2$. We can construct $H'_{\nu,\rho}$ by merging in parallel two representations $H'_{\mu_l,\rho}$ of $G_{\mu_l,\rho}$ and $H'_{\mu_r,\rho}$ of $G_{\mu_r,\rho}$ (in the same left-to-right order they have in $H_{\nu,\rho}$) in such a way that: $H'_{\mu_l,\rho}$ has spirality $\sigma'_{\mu_l}=\sigma_{\mu_l}-2$, $\sigma'_{\mu_r}=\sigma_{\mu_r}$, ${\alpha'}_u^l={\alpha'}_v^l=0$, and ${\alpha'}_u^r={\alpha'}_v^r=1$. This way, $H_{\nu,\rho}$ has spirality $\sigma'_\nu=\sigma_\nu$ and $\sigma'_{\mu_l}-\sigma'_{\mu_r} = 2$.  See Fig.~\ref{fi:spir-nodiff4-e}, where $\sigma_\nu=2$.
\end{proof}

\begin{figure}[tb]
	\centering
	\subfigure[]{\label{fi:spir-nodiff4-a}\includegraphics[width=0.3\columnwidth,page=1]{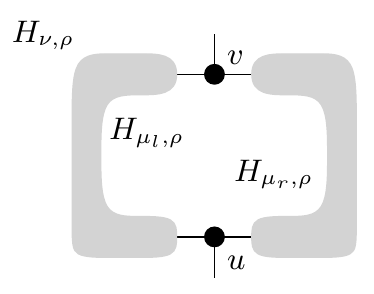}}
	\hfil
	\subfigure[]{\label{fi:spir-nodiff4-b}\includegraphics[width=0.3\columnwidth,page=2]{spir-nodiff4}}
	\hfil
	\subfigure[]{\label{fi:spir-nodiff4-c}\includegraphics[width=0.3\columnwidth,page=3]{spir-nodiff4}}
	\hfil
	\subfigure[]{\label{fi:spir-nodiff4-d}\includegraphics[width=0.3\columnwidth,page=4]{spir-nodiff4}}
	\hfil
	\subfigure[]{\label{fi:spir-nodiff4-e}\includegraphics[width=0.3\columnwidth,page=5]{spir-nodiff4}}
	\caption{Illustration for the proof of \cref{le:intervalSupportFirst}.}\label{fi:external-face}
\end{figure}

\begin{restatable}{lemma}{leIntervalsSupportThird}\label{le:intervalSupportThird}
	Let $\Sigma^+_{\nu,\rho}$ be a non-trivial interval with maximum value $M > 2$. If $\Sigma^+_{\nu,\rho}$ contains an integer with parity different from that of $M$, ~$\Sigma^+_{\nu,\rho}=[0,M]^1$.
\end{restatable}

\begin{proof}
	Assume that $M$ is odd (if $M$ is even, the proof is similar). By hypothesis $M \geq 3$. 
	%Suppose that $\Sigma_{\nu,\rho}^+$ contains a value whose parity is different from the parity of $M$. 
	We prove that, if $\Sigma^+_\nu$ contains a value $\sigma_\nu$ whose parity is different from the one of $M$, then $\Sigma_{\nu,\rho}^+=[0,M]^1$. The proof is by induction on the depth of the subtree of $T_\rho$ rooted at $\nu$.
	If $\nu$ is a Q$^*$-node, then $\Sigma_{\nu,\rho}^+=[0,M]^1$ and the statement trivially holds. In the inductive case,
	$\nu$ is either an S-node or a P-node. By Corollary~\ref{co:intervalSupport}, $G_{\nu,\rho}$ admits spirality $\sigma'_\nu$ for every $\sigma'_\nu \in [1,M]^2$.
	Below, we analyze separately the case when $\nu$ is an S-node, a P-node with three children, or a P-node with two children.
	
	\smallskip \noindent
	\textsf{-- $\nu$ is an S-node.}    
	We prove that for any value $\sigma'_\nu \in [1,M]^2$, $G_{\nu,\rho}$ also admits spirality $\sigma'_\nu - 1$. This immediately implies that $\Sigma_{\nu,\rho}^+=[0,M]^1$. We first prove the following claim:
	\begin{claim}
		There exists a child $\mu$ of $\nu$ in $T_{\rho}$ that is jump-1.
	\end{claim}
	\begin{claimproof}
		Let $H_{\nu,\rho}$ be a representation of $G_{\nu,\rho}$ with spirality $\sigma_\nu$ and let $H_{\nu,\rho}'$ be a representation of $G_{\nu,\rho}$ with spirality $\sigma'_\nu = \sigma_\nu+1$. Note that $\sigma'_\nu \in [1,M]^2$, thus
		$H_{\nu,\rho}'$ exists. By Lemma~\ref{le:spirality-S-node}, 
		since the spiralities of $H_{\nu,\rho}$ and of $H'_{\nu,\rho}$ have different parities, $\nu$ must have a child $\mu$ such that $H_{\mu,\rho}$ has odd spirality in $H_{\nu,\rho}$ and even spirality in $H'_{\nu,\rho}$, or vice versa. 
		Let $M_\mu$ be the maximum spirality admitted by $\mu$. Since $\mu$ admits both an even and an odd value of spirality, we have: If $M_\mu=1$, $\mu$ admits~$0$ and $\Sigma_{\mu,\rho}^+=[0,1]^1$; if $M_\mu= 2$, by Property~(a) of Lemma~\ref{le:intervalSupportFirst} and since $\mu$ admits spirality $1$, either $\Sigma_{\mu,\rho}^+=[0,2]^1$ or $\Sigma_{\mu,\rho}^+=[1,2]^1$; if $M_\mu>2$, by inductive hypotesis $\Sigma_{\mu,\rho}^+= [0,M_\mu]^1$. Hence, $\mu$ is always jump-1. 
	\end{claimproof}
	
	Let $\mu$ be a child of $\nu$ having a jump-1 interval, which always exists by the previous claim. For any value $\sigma'_\nu \in [1,M]^2$, let $H'_{\nu,\rho}$ be a rectilinear representation of $G_{\nu,\rho}$ with spirality $\sigma'_\nu$. 
	Let $\sigma_\mu$ be the spirality of the restriction of~$H'_{\nu,\rho}$ to~$G_{\mu,\rho}$. 
	Suppose first that $\sigma_\mu>-M_\mu$. Since by inductive hypothesis $\mu$ admits spirality $\sigma_\mu-1$ then, by Lemma~\ref{le:spirality-S-node}, $\nu$ admits $\sigma'_\nu-1$.
	Suppose now that $\sigma_\mu=-M_\mu$. Since $\sigma'_\nu > 0$, by Lemma~\ref{le:spirality-S-node}, there exists a child $\phi \neq \mu$ of $\nu$ such that the restriction of $H'_{\nu,\rho}$ to $G_{\phi,\rho}$ has spirality $\sigma_{\phi}> 0$.
	Observe that $\phi$ also admits either spirality $\sigma_{\phi}-1$ or spirality $\sigma_{\phi}-2$. Indeed, if $\sigma_{\phi}> 2$, then $\phi$ admits spirality $\sigma_{\phi}-2$ by Property~(b) of Lemma~\ref{le:intervalSupportFirst}; if $\sigma_{\phi}=2$ it also admits spirality 0 or 1 by Property~(a) of Lemma~\ref{le:intervalSupportFirst}; if $\sigma_{\phi}=1$ then it also admits spirality -1.    		
	In the case that $\phi$ admits spirality $\sigma_{\phi}-1$, by Lemma~\ref{le:spirality-S-node}, $\nu$ admits spirality $\sigma'_\nu-1$. In the case that $\phi$ admits spirality $\sigma_{\phi}-2$, then $\mu$ admits spirality $\sigma_\mu+1$ (because $\mu$ is jump-1 and we are assuming  $\sigma_\mu=-M_\mu < M_\mu$), and hence $\nu$ admits spirality $\sigma_\phi-2+1=\sigma'_\nu-1$.
	
	\smallskip\noindent
	\textsf{-- $\nu$ is a P-node with three children.} 
	In this case every child $\mu$ of $\nu$ is jump-1.
	Indeed, since $\nu$ admits an even and an odd value of spirality, by Lemma~\ref{le:spirality-P-node-3-children}, the same holds for $\mu$.  
	As for the case of an S-node, if $M_\mu$ is the maximum value of spirality admitted by $\nu$, we have the following: If $M_\mu=1$, $\Sigma_\mu^+=[0,1]^1$; if $M_\mu= 2$, either $\Sigma_{\mu,\rho}^+=[0,2]^1$ or $\Sigma_{\mu,\rho}^+=[1,2]^1$; if $M_\mu>2$, by inductive hypotesis $\Sigma_{\mu,\rho}^+= [0,M_\mu]^1$. Hence, $\mu$ is jump-1.
	
	Assume first that $M > 3$. Let $H_{\nu,\rho}$ be a representation of $G_{\nu,\rho}$ with spirality~$M$. By Lemma~\ref{le:spirality-P-node-3-children}, every child $\mu$ of $\nu$, is such that the restriction of $H_{\nu,\rho}$ to $G_{\mu,\rho}$ has spirality $\sigma_{\mu} \geq 2$. Since $\mu$ is jump-1, then $\mu$ also admits spirality $\sigma_{\mu} - 1$. This implies that, $\nu$ admits a representation with spirality $M-1$. Since $M-1>2$, by Corollary~\ref{co:intervalSupport}, $\nu$ admits all values of spirality in the set $[0,M-1]^2$, and hence  
	$\Sigma_{\nu,\rho}^+=[0,M-1]^2 \cup [1,M]^2 = [0,M]^1$.
	
	Assume now that $M=3$. Let $H_{\nu,\rho}$ be a representation of $G_{\nu,\rho}$ with spirality~$M$. The restrictions of $H_{\nu,\rho}$ to the three child components $G_{\mu_l,\rho}$, $G_{\mu_c,\rho}$, and $G_{\mu_r,\rho}$ of $G_{\nu,\rho}$, have spiralities 5, 3, and 1, respectively. Since $\mu_l$ is jump-1, by the inductive hypothesis it admits spirality for all values in the set $[0,5]^1$. Similarly, since $\mu_c$ is jump-1, by the inductive hypothesis it admits spirality for all values in the set $[0,3]^1$.  
	Also, since $\mu_r$ is jump-1, it admits spirality 0 or 2. If $\mu_r$ admits spirality 0, then $\nu$ admits spirality $M-1=2$ for a representation in which $G_{\mu_l,\rho}$, $G_{\mu_c,\rho}$, and $G_{\mu_r,\rho}$ appear in this left-to-right order (and have spiralities 4, 2, and 0, respectively). If $\mu_r$ admits spirality 2 but not spirality 0, then $\nu$ admits spirality $M-1=2$ for a representation in which $G_{\mu_l,\rho}$, $G_{\mu_r,\rho}$, and $G_{\mu_c,\rho}$ appear in this order (and again have spiralities 4, 2, and 0, respectively). Hence, so far we have proved that $\nu$ admits spirality for all values in the set $[1,3]^1$. Finally, as showed in the proof of Property~(a) of Lemma~\ref{le:intervalSupportFirst} for a P-node with three children, the fact that $\nu$ admits spirality 2 implies that it also admits spirality 0 (see Figs.~\ref{fig:3children2-0-a} and~\ref{fig:3children2-0-d}).  
	
	\smallskip \noindent
	\textsf{-- $\nu$ is a P-node with two children.} Let $H_{\nu,\rho}$ be a rectilinear planar representation of $G_{\nu,\rho}$ with spirality $M$. Let $\sigma_{\mu_l}$ and $\sigma_{\mu_r}$ be the spiralities of the restrictions of $H_{\nu,\rho}$ to the left and right child components $G_{\mu_l,\rho}$ and $G_{\mu_r,\rho}$ of $G_{\nu,\rho}$, respectively. Also, let $\{u,v\}$ be the poles of $\nu$.
	By Lemma~\ref{le:intervalSupportSecond}, we can assume $\sigma_{\mu_l}-\sigma_{\mu_r}\in [2,3]$, which implies that there exists $w\in \{u,v\}$ such that $\alpha_w^l=0$, as $H_{\nu,\rho}$ has the maximum value of spirality admitted by $\nu$. By Lemma~\ref{le:spirality-P-node-2-children}, for $\alpha_w^l=1$ and $\alpha_w^r=0$ we can obtain a rectilinear planar representation of $G_{\nu,\rho}$ with spirality $M-1$.
	If $M>3$ then $M-1>2$ and, by Corollary~\ref{co:intervalSupport}, $\nu$ admits spirality for all values in the set $[0,M-1]^2$, and hence $\Sigma_{\mu,\rho}^+=[0,M-1]^2 \cup [1,M]^2 = [0,M]^1$. 
	If $M=3$, by Property~(a) of Lemma~\ref{le:intervalSupportFirst}, we have either $[0,3]^1\in \Sigma_{\mu,\rho}^+$ or $[1,3]^1\in \Sigma_{\mu,\rho}^+$. In the former case, $\Sigma_{\mu}^+=[0,M]^1$. In the latter case, using a case analysis similar to the proof of Property~(a) of Lemma~\ref{le:intervalSupportFirst} for the P-nodes with two children, it can be proved that 0 is also admitted by $\nu$, and again $\Sigma_{\mu}^+=[0,M]^1$.  
	%	
	%		one can show that 0 is also admitted by $\nu$. In the proof of Property~(a) Lemma~\ref{le:intervalSupportFirst} for the two children we showed that, given a representation of $G_{\nu,\rho}$ where $\sigma_\nu=2$, there exists also a representation of $G_{\nu,\rho}$ where $\sigma_\nu=0$ for Cases~(i) (Figure~\ref{fig:2children2-0-a} and~\ref{fig:2children2-0-b}), Case~(ii) (Figure~\ref{fig:2children2-0-c} and~\ref{fig:2children2-0-d}), Case~(iii) (Figure~\ref{fig:2children2-0-e} and~\ref{fig:2children2-0-f}), and Case~(v) (Figure~\ref{fig:2childrenDiffspir34-e} and~\ref{fig:2childrenDiffspir34-f}). It remains to show that by adding the hypothesis that $\nu$ admits $1$ and $3$, which was not an hypothesis of Lemma~\ref{le:intervalSupportFirst}, $\nu$ admits $0$ also for Case~(iv), depicted in Figure~\ref{fig:2childrenDiffspir34-a}.
	%		
	%		Since $\nu$ admits a representation with spirality $3$ we have that, for any child $\mu$ of $\nu$, $\mu$ admits  $1$. Hence, given the representation of Case~(iv) of Figure~\ref{fig:2childrenDiffspir34-a}, it is sufficient to set the spiralities and the angles around the poles as in Figure~\ref{fig:2children2-0-d}. Namely: $\sigma_{\mu_l}=1$, which is possible by Property~(a) of Lemma~\ref{le:spirality-P-node-2-children}; $\sigma_{\mu_r}=-1$. By Lemma~\ref{le:spirality-P-node-2-children} we have that the correspondent representation have spirality $\sigma_\nu=0$ and, consequently, also in this case we have $\Sigma_\mu^+=[0,M]^1$.
\end{proof}

\myparagraph{Proof of Theorem~\ref{th:spirality-sets}.}	
	Let $M$ be the maximum value in $\Sigma^+_{\nu,\rho}$. If $M=0$ then $\Sigma^+_{\nu,\rho} = [0]$. If $M=1$ then either $\Sigma^+_{\nu,\rho} = [1]$ or $\Sigma^+_{\nu,\rho} = [0,1]^1 = [0,M]^1$. Suppose $M=2$; by Property~(a) of Lemma~\ref{le:intervalSupportFirst}, $G_{\nu,\rho}$ admits spirality 0, or 1, or both, i.e., $\Sigma^+_{\nu,\rho} = [0,2]^2=[0,M]^2$, or $\Sigma^+_{\nu,\rho} = [1,2]^1$, or $\Sigma^+_{\nu,\rho} = [0,2]^1=[0,M]^1$. Finally, suppose that $M > 2$. If $G_{\nu,\rho}$ admits a value of spirality whose parity is different from $M$, by Lemma~\ref{le:intervalSupportThird} $\Sigma^+_{\nu,\rho}=[0,M]^1$; else, by Corollary~\ref{co:intervalSupport}, either $\Sigma^+_{\nu,\rho}=[1,M]^2$ (if $M$ is odd) or $\Sigma^+_{\nu,\rho}=[0,M]^2$ (if $M$ is even).

\smallskip\noindent Figure~\ref{fi:intervals-new} shows examples of the non-negative spirality sets in Theorem~\ref{th:spirality-sets}.

%SUBTOREMOVE
%\subsection{Rectilinear Spirality Sets of Q$^*$-, S-, and P-nodes}\label{sse:more-on-sets}
%In this section we give additional properties of the rectilinear spirality sets of Q$^*$-nodes, S-nodes, and P-nodes, which are used to design our testing algorithm.

%%%% TEST
%\clearpage
\section{Rectilinear Planarity Testing}\label{se:testing}
%\begin{lemma}[Reusability Principle]\label{le:reusability}
%	Let $T=(V_T,E_T)$ be a decomposition tree of an $n$-vertex graph such that $T$ has size $O(n)$. Let $V_R =\{r_1, r_2, \dots, r_h\} \subseteq V_T$ be a set of nodes of $T$ and let $T_{r_1}, T_{r_2}, \dots, T_{r_h}$ be a sequence of trees obtained by rooting ${T}$ at the nodes in $V_R$.
%	Let $\mathcal{A}$ be an algorithm such that: (i) for each $r_i$, $1 \leq i \leq h$, $\mathcal{A}$ performs a bottom-up traversal of $T_{r_i}$; (ii) for every node $v \in V_T$ with $k$ children in $T_{r_i}$, $\mathcal{A}$ computes a value $\textsc{val}_{r_i}(v)$ in $O(k)$ time if $i=1$ and in $O(1)$ time for $2 \leq i \leq h$. There exists an algorithm $\mathcal{A^+}$ that computes the set $\{\textsc{val}_{r}(r)| r \in V_R\}$ in $O(n)$ time.
%\end{lemma}
%Our rectilinear planarity testing algorithm is based on dynamic programming (Theorem~\ref{th:timeTest}).
Let $G$ be a biconnected \pisp that is not a simple cycle,~$T$ be its SPQ$^*$-tree, and $\{\rho_1, \dots, \rho_h\}$ be the Q$^*$-nodes of~$T$. For each possible choice of the root $\rho \in \{\rho_1, \dots, \rho_h\}$, the algorithm visits $T_\rho$ bottom-up in post-order and computes, for each visited node $\nu$, the non-negative spirality set  $\Sigma^+_{\nu,\rho}$, based on the sets of the children of $\nu$. $\Sigma^+_{\nu,\rho}$ is representative of all ``shapes'' that $G_{\nu,\rho}$ can take in a rectilinear planar representation of $G$ with the reference chain on the external face. The key lemmas used to show that we can efficiently execute this procedure over all SPQ$^*$-tree $T_\rho$ of~$G$ are Lemmas~\ref{le:timeQ}, \ref{le:timeS}, \ref{le:timeP}, and \ref{le:timeRoot}.
%They are proved in the next subsections. In all the statements we assume that $G$ is an \pisp and $T_\rho$ is the SPQ$^*$-tree of $G$ rooted at a Q$^*$-node $\rho$. Also, we assume that for a node $\nu$ of $T_\rho$, all child components of $G_{\nu,\rho}$ are rectilinear planar.
%
%The proofs of Lemmas~\ref{le:timeS} and~\ref{le:timeP} rely on some technical lemmas, given in the following, where we treat separately the case of an S-node, of a P-node with three children and of a P-node with two children. In all the statements we assume that $G$ is an \pisp and $T_\rho$ is the SPQ$^*$-tree of $G$ rooted at a Q$^*$-node $\rho$. Also, we assume that for a node $\nu$ of $T_\rho$, all child components of $G_{\nu,\rho}$ are rectilinear planar.
%
From now on, we say that a node $\nu$ in $T_\rho$ is \emph{trivial}, or \emph{jump-1}, or \emph{jump-2}, if $\Sigma^+_{\nu,\rho}$ is a trivial interval, or a jump-1 interval, or a jump-2 interval, respectively.

%%%% ALGORITHMIC

\smallskip
\myparagraph{Q$^*$-nodes.} Each chain of length $\ell$ can turn at most $\ell-1$ times (one turn for each vertex). Therefore, for a Q$^*$-node $\nu$ of $T_\rho$, we have $\Sigma^+_{\nu,\rho} = [0,\ell-1]^1$, and the following lemma holds, assuming that each Q$^*$-node is equipped with the length of its corresponding chain when we compute the SPQ$^*$-tree $T$ of~$G$.
%Since a chain of edges of length $\ell$ can turn at most $\ell-1$ times (one turn for each vertex), we immediately have the following.

%\begin{restatable}{lemma}{leQIntervalSupport}\label{le:QIntervalSupport}
%	Let $\nu \neq \rho$ be a Q$^*$-node of $T_\rho$. Node $\nu$ is jump-1 and $\Sigma^+_{\nu,\rho}=[0,\ell-1]$, where $\ell$ is the length of the chain represented by $\nu$.
%\end{restatable}
%
%The next lemma is an immediate consequence of Lemma~\ref{le:QIntervalSupport}, where we assume that each Q$^*$-node is equipped with the length of the chain it represents when we compute the SPQ$^*$-tree of $G$.

\begin{restatable}{lemma}{letimeQ}\label{le:timeQ}
	Let $G$ be an \pisp, $T_\rho$ be a rooted SPQ$^*$-tree of $G$, and $\nu$ be a Q$^*$-node of $T_{\rho}$. The set $\Sigma^+_{\nu,\rho}$ can be computed in $O(1)$~time.
\end{restatable}
%\begin{proof}
%	Each chain of edges of length $\ell$ can turn at most $\ell-1$ times (one turn for each vertex). The lemma follows assuming that each Q$^*$-node is equipped with the length of the chain it represents when we compute the SPQ$^*$-tree of $G$.
%\end{proof}

\myparagraph{S-nodes.} To prove Lemma~\ref{le:timeS} we first state a property of S-nodes of \pisps.%, whose proof can be found in App.~\ref{app:testing}.

\begin{restatable}{lemma}{leSIntervalSupport}\label{le:SIntervalSupport}
	Let $\nu$ be an S-node of $T_\rho$. Node $\nu$ is jump-1 if and only if at least one of its children is jump-1. Also, $\Sigma^+_{\nu,\rho}=[1,2]^1$ if and only if $\nu$ has exactly one child with non-negative rectilinear spirality set $[1,2]^1$ and all the other children with non-negative rectilinear spirality set $[0]$.
\end{restatable}
\begin{proof}
	%\textcolor{red}{Use the proof of Lemma 10 of Giacomo for the first part; also use the proof of Lemma 11 of Giacomo for the secon part.}
	We prove that $\nu$ is jump-1 if and only if at least one of its children is jump-1. Suppose first that $\nu$ is jump-1 and suppose by contradiction that all its children are trivial or jump-2. This implies that for each child $\mu$ of $\nu$, $\Sigma^+_{\mu,\rho}$ contains only even values or only odd values. Denote by $j$ the number of children of $\nu$ whose non-negative rectilinear spirality set contain only odd values. By Lemma~\ref{le:spirality-S-node}, the spirality of any rectilinear representation of $G_{\nu,\rho}$ is the sum of the spiralities of all child components. It follows that $G_{\nu,\rho}$ admits only even values of spiralities if $j$ is even and only odd values of spiralities if $j$ is odd, which contradicts the hypothesis that $\nu$ is jump-1. Suppose vice versa that $\nu$ has at least a child $\mu$ that is jump-1. Denote by $M$ the maximum value in $\Sigma^+_{\nu,\rho}$ and by $M_\mu$ the maximum value in $\Sigma^+_{\mu,\rho}$. Let $H_{\nu,\rho}$ be any rectilinear planar representation of $G_{\nu,\rho}$ having spirality $M$, and let $H_{\mu,\rho}$ be its restriction to $G_{\mu,\rho}$. By  Lemma~\ref{le:spirality-S-node}, $H_{\mu,\rho}$ has spirality $M_\mu$. Also, since $\mu$ is jump-1, by Lemma~\ref{le:spirality-S-node} we can obtain a rectilinear representation $H'_{\nu,\rho}$ of $G_{\nu,\rho}$ with spirality $M-1$ by simply replacing $H_{\mu,\rho}$ in $H_{\nu,\rho}$ with a rectilinear representation of $G_{\mu,\rho}$ having spirality $M_\mu - 1$. Therefore, by Theorem~\ref{th:spirality-sets}, $\nu$ is jump-1.
	
	We now show the second part of the lemma. Suppose first that $\nu$ has exactly one child $\mu$ with non-negative rectilinear spirality set $[1,2]^1$ and all the other children with non-negative rectilinear spirality set $[0]$. Clearly, by Lemma~\ref{le:spirality-S-node}, $\Sigma^+_{\nu,\rho}=\Sigma^+_{\mu,\rho}$, i.e., $\Sigma^+_{\nu,\rho}=[1,2]^1$. Suppose vice versa that $\Sigma^+_{\nu,\rho}=[1,2]^1$. By Lemma~\ref{le:spirality-S-node}, the sum of the spiralities admitted by the child components of $\nu$ cannot be larger than two. If exactly one child of $\nu$ has non-negative rectilinear spirality set $[1,2]^1$ and all the other children have non-negative rectilinear spirality set $[0]$, we are done. Otherwise, one of the following two cases must be considered: (i)~There are two children $\mu$ and $\mu'$ of $\nu$ such that the maximum value of spirality admitted by $G_{\mu,\rho}$ and $G_{\mu',\rho}$ is $1$ and any other child $\nu$ has non-negative rectilinear spirality set $[0]$; this case is ruled out by observing that $G_{\mu,\rho}$ (and $G_{\mu',\rho}$) would also admit spirality $-1$ and thus, by Lemma~\ref{le:spirality-S-node}, $G_{\nu,\rho}$ would also admit spirality~0. (ii)~$\nu$ has a child $\mu$ for which either $\Sigma_{\mu,\rho}^+=[0,2]^1$ or  $\Sigma_{\mu,\rho}^+=[0,2]^2$ and any other child of $\nu$ has non-negative rectilinear spirality set $[0]$; again, this case is ruled out because it would imply that also $G_{\nu,\rho}$ admits spirality~0. 	
\end{proof}
\begin{restatable}{lemma}{letimeS}\label{le:timeS}
	Let $G$ be an \pisp, $T$ be the SPQ$^*$-tree of~$G$, $\nu$ be an S-node of~$T$ with $n_\nu$ children, and $\rho_1, \rho_2, \dots, \rho_h$ be the Q$^*$-nodes of~$T$.
	Assume that, for each child $\mu$ of $\nu$ in $T_{\rho_i}$, $\Sigma^+_{\mu,\rho_i}$ is given. The set $\Sigma^+_{\nu,\rho_i}$ can be computed in $O(n_\nu)$ time for $i = 1$ and in $O(1)$ time for $2 \leq i \leq h$.
\end{restatable}
\begin{proof}
	For any $i = 1, \dots, h$, let $x_{\nu,\rho_i}$ and $y_{\nu,\rho_i}$ be the number of children of~$\nu$ in $T_{\rho_i}$ with non-negative spirality set $[0]$ and $[1,2]^1$, respectively. Also, let $z_{\nu,\rho_i}$ be the number of children that are jump-1 (clearly, $z_{\nu,\rho_i} \geq y_{\nu,\rho_i}$). Let $M_{\nu,\rho_i}$ be the maximum value in $\Sigma^+_{\nu,\rho_i}$.
	First, we show how to compute $\Sigma^+_{\nu,\rho_i}$ in $O(1)$ time given $x_{\nu,\rho_i}$, $y_{\nu,\rho_i}$, $z_{\nu,\rho_i}$, $M_{\nu,\rho_i}$.
	By Lemma~\ref{le:SIntervalSupport}, $\Sigma_{\nu,\rho_i}^+$ is jump-1 if and only if $z_{\nu,\rho_i}>0$. Suppose that $\Sigma_{\nu,\rho_i}^+$ is jump-1. If $M_{\nu,\rho_i}\not =2$, by Theorem~\ref{th:spirality-sets},  $\Sigma_{\nu,\rho_i}^+=[0,M_{\nu,\rho_i}]^1$ . If $M_{\nu,\rho_i} =2$,  Lemma~\ref{le:SIntervalSupport} implies $\Sigma_{\nu,\rho_i}^+=[1,2]^1$ if $x_{\nu,\rho_i}+y_{\nu,\rho_i}=n_\nu$ and $y_{\nu,\rho_i}=1$; otherwise  $\Sigma_{\nu,\rho_i}^+=[0,2]^1$.
%	if $x_{\nu,\rho_i}+y_{\nu,\rho_i}=n_\nu$ and $y_{\nu,\rho_i}=1$, we have $\Sigma_{\nu,\rho_i}^+=[1,2]^1$; else, $\Sigma_{\nu,\rho_i}^+=[0,2]^1$.
	Suppose now that  $\Sigma_{\nu,\rho_i}^+$ is not jump-1. By Theorem~\ref{th:spirality-sets} we have: $\Sigma_{\nu,\rho_i}^+=[0]$ if $M_{\nu,\rho_i}=0$ and $\Sigma_{\nu,\rho_i}^+=[1]$ if $M_{\nu,\rho_i}=1$; $\Sigma_{\nu,\rho_i}^+=[1,M_{\nu,\rho_i}]^2$ if $M_{\nu,\rho_i}>1$ and $M_{\nu,\rho_i}$ is odd; $\Sigma_{\nu,\rho_i}^+=[0,M_{\nu,\rho_i}]^2$ if $M_{\nu,\rho_i}>1$ and~$M_{\nu,\rho_i}$ is~even.
	
	We now show ho to compute $x_{\nu,\rho_i}$, $y_{\nu,\rho_i}$, and $z_{\nu,\rho_i}$ for $i=1, \dots, h$. If $i=1$, given $\Sigma^+_{\mu,\rho_1}$ for every child $\mu$ of $\nu$ in~$T_{\rho_1}$, then $x_{\nu,\rho_1}$, $y_{\nu,\rho_1}$, and $z_{\nu,\rho_1}$ are computed in $O(n_\nu)$ time by just visiting each child of $\nu$. Also, since by Lemma~\ref{le:spirality-S-node} the maximum spirality admitted by $G_{\nu,\rho}$ is the sum of the maximum spiralities admitted by the children of $\nu$ in $T_{\rho_1}$, we also compute $M_{\nu,\rho_1}$ and $\Sigma^+_{\nu,\rho_1}$ in $O(n_\nu)$ time. We store at $\nu$ the values $x_{\nu,\rho_1}$, $y_{\nu,\rho_1}$, $z_{\nu,\rho_1}$, and~$M_{\nu,\rho_1}$.
	
	Let $i \in \{2, \dots, h\}$. Let $\mu_1$ be the parent of $\nu$ in $T_{\rho_1}$ and let $\mu_i$ be the parent of $\nu$ in $T_{\rho_i}$. Note that, $\mu_1$ is a child of $\nu$ in~$T_{\rho_i}$ and $\mu_i$ is a child of $\nu$ in~$T_{\rho_1}$. Any other child of $\nu$ in $T_{\rho_1}$ is also a child of $\nu$ in $T_{\rho_i}$ and vice versa.
	To compute $\Sigma^+_{\nu,\rho_i}$ in $O(1)$ time, we compute~$x_{\nu,\rho_i}$, ~$y_{\nu,\rho_i}$, ~$z_{\nu,\rho_i}$,~$M_{\nu,\rho_i}$ as follows:
	%
	%\smallskip\noindent
	(i)~Let $g_{\mu_i}=1$ if $\Sigma^+_{\mu_i,\rho_1}=[0]$ and $g_{\mu_i}=0$ otherwise. Also, let $g_{\mu_1}=1$ if  $\Sigma_{\mu_1,\rho_i}=[0]$ and $g_{\mu_1}=0$ otherwise. We have $x_{\nu,\rho_i}=x_{\nu,\rho_1}-g_{\mu_i}+g_{\mu_1}$.
	%
	%\smallskip\noindent
	(ii)~Let $g'_{\mu_i}=1$ if $\Sigma^+_{\mu_i,\rho_1}=[1,2]^1$ and $g'_{\mu_i}=0$ otherwise. Also, let $g'_{\mu_1}=1$ if  $\Sigma_{\mu_1,\rho_i}=[1,2]^1$ and  $g_{\mu_1}=0$ otherwise. We have $y_{\nu,\rho_i}=y_{\nu,\rho_1}-g'_{\mu_i}+g'_{\mu_1}$.
	%	
	%\smallskip\noindent
	(iii)~Let $g''_{\mu_i}=1$ if $\Sigma_{\mu_i,\rho_1}$ is jump-1 and $g''_{\mu_i}=0$ otherwise. Also, let $g''_{\mu_1}=1$ if  $\Sigma^+_{\mu_1,\rho_i}$ is jump-1 and  $g''_{\mu_1}=0$ otherwise. We have $z_{\nu,\rho_i}=z_{\nu,\rho_1}-g''_{\mu_i}+g''_{\mu_1}$.
	%	
	%\smallskip\noindent
	(iv)~$M_{\nu,\rho_i}=M_{\nu,\rho_1}-M_{\mu_i,\rho_1}+M_{\mu_1,\rho_i}$.
\end{proof}

%\subsection{P-nodes}\label{sse:P-nodes}
%\myparagraph{P-nodes with three children.}

%\begin{restatable}{lemma}{leP3CheckSingleValue}\label{le:P3CheckSingleValue}
%	Let $\nu$ be a P-node of $T_\rho$ with three children $\mu_1$, $\mu_2$, and $\mu_3$. Assume that $\Sigma^+_{\mu_i,\rho}$ is given and $\Sigma^+_{\mu_i,\rho} \neq \emptyset$ for each $1 \leq i \leq 3$. For any given integer $\sigma_\nu \geq 0$, it is possible to check in $O(1)$ time whether $\sigma_\nu \in \Sigma^+_{\nu,\rho}$.
%\end{restatable}
%\begin{proof}
%	\textcolor{red}{TO DO - Lemma 14 of Giacomo - easy and short}
%\end{proof}

%\begin{restatable}{lemma}{leP3CheckSingleValue}\label{le:P3CheckSingleValue}
%	Let $\nu$ be a P-node of $T_\rho$ with three children and assume that the non-negative rectilinear spirality set of every child of $\nu$ is given. For any given integer $\sigma_\nu \geq 0$, it is possible to check in $O(1)$ time whether $\sigma_\nu \in \Sigma^+_{\nu,\rho}$.
%\end{restatable}
%\begin{proof}
%	\textcolor{red}{TO DO - Lemma 14 of Giacomo - easy and short}
%\end{proof}

\myparagraph{P-nodes.} For a P-node $\nu$, $\Sigma^+_{\nu,\rho}$ can be computes in $O(1)$ time, independent of~$\rho$.

\begin{restatable}{lemma}{leTimePThree}\label{le:timeP}
	Let $G$ be an \pisp and let $T_\rho$ be a rooted SPQ$^*$-tree of $G$. Let $\nu$ be a P-node of~$T_{\rho}$ and assume that, for each child $\mu$ of~$\nu$ in~$T_{\rho}$, the set $\Sigma^+_{\mu,\rho}$ is given. The set $\Sigma^+_{\nu,\rho}$ can be computed in $O(1)$ time.
\end{restatable}

To prove Lemma~\ref{le:timeP}, we prove Lemmas~\ref{le:timeP3} and~\ref{le:timeP2}, which treat separately the case of a P-node with three children and the case of a P-node with two children.

%%%%%%% P-NODES WITH 3 CHILDREN

%\leTimePThree*

\begin{restatable}{lemma}{leTimePThree}\label{le:timeP3}
	Let $G$ be an \pisp and let $T_\rho$ be a rooted SPQ$^*$-tree of $G$. Let $\nu$ be a P-node of $T_{\rho}$ with three children and assume that, for each child $\mu$ of $\nu$ in $T_{\rho_i}$, the set $\Sigma^+_{\mu,\rho_i}$ is given.
	The set $\Sigma^+_{\nu,\rho}$ can be computed in $O(1)$ time.
\end{restatable}
\begin{proof}
	Observe that, by Lemma~\ref{le:spirality-P-node-3-children}, for any given integer value $\sigma_\nu \geq 0$, one can test in $O(1)$ time whether $G_{\nu,\rho}$ admits spirality~$\sigma_\nu$. It suffices to test if there exists a child of $\nu$ that admits spirality $\sigma_\nu$, another child that admits spirality $\sigma_\nu+2$, and the remaining child that admits spirality $\sigma_\nu-2$. Testing this condition requires a constant number of checks. 
	
	By Theorem~\ref{th:spirality-sets}, $G_{\nu,\rho}$ is rectilinear planar if and only if it admits spirality either $0$ or $1$. Based on the previous observation, we can check this property in $O(1)$ time; if it does not hold, then $\Sigma^+_{\nu,\rho}=\emptyset$. Otherwise, we determine the maximum value $M$ in $\Sigma^+_{\nu,\rho}$. By Theorem~\ref{th:spirality-sets}, it suffices to find a value $\sigma_\nu$ such that $\nu$ admits spirality $\sigma_\nu$ but it does not admit spirality $\sigma_\nu + 1$ and $\sigma_\nu + 2$; if we find such a value, then $M=\sigma_\nu$. Using this observation, we prove that one can find $M$ in $O(1)$ time.  
	
	For each $i=0,...,4$, we can first check in $O(1)$ time whether $M=i$. If this is not the case, then $M>4$. To find $M$ in this case, we first claim an interesting property.  
	Consider the maximum values in the non-negative rectilinear spirality sets of the children of $\nu$. Denote by $\mu_{\max}$ (resp. $\mu_{\min}$) any child of $\nu$ whose maximum value is not smaller than (resp. the larger than) any other maximum values. Also denote by $\mu_{\md}$ the remaining child. We prove the following claim.
	
	%	\begin{claim}
	%		Let $\nu$ be a P-node of $T_\rho$ with children $\mu_{\max}$, $\mu_{\md}$, and $\mu_{\min}$. Suppose that $G_{\nu,\rho}$ is rectilinear planar and let $M$ be the maximum value in $\Sigma^+_{\nu,\rho}$. If $M > 4$, there exists a rectilinear planar representation of $G_{\nu,\rho}$ with spirality $M$ whose embedding is such that $G_{\mu_{\max},\rho}$, $G_{\mu_{\md},\rho}$, and $G_{\mu_{\min},\rho}$ appear from left to right in this order.
	%	\end{claim}
	%\clPThree*
	\begin{restatable}{claim}{clPThree}
		Let $M$ be the maximum value in $\Sigma^+_{\nu,\rho}$. If $M > 4$ then $G_{\nu,\rho}$ admits spirality $M$ for an embedding where $G_{\mu_{\max},\rho}$, $G_{\mu_{\md},\rho}$, and $G_{\mu_{\min},\rho}$ appear in this left-to-right order.
	\end{restatable}  
	\begin{claimproof}
		Let $H_{\nu,\rho}$ be a rectilinear planar representation of $G_{\nu,\rho}$ with spirality $M > 4$. If $G_{\mu_{\max},\rho}$, $G_{\mu_{\md},\rho}$, and $G_{\mu_{\min},\rho}$ appear in this order in $H_{\nu,\rho}$ we are done. Hence, suppose this is not the case; we prove that there exists another rectilinear planar representation $H'_{\nu,\rho}$ of $G_{\nu,\rho}$ with spirality $M$ and such that $G_{\mu_{\max},\rho}$, $G_{\mu_{\md},\rho}$, and $G_{\mu_{\min},\rho}$ appear in this left-to-right order in the planar embedding of $H'_{\nu,\rho}$.
		
		Let $\mu_l$, $\mu_c$, and $\mu_r$ be the children of $\nu$ that correspond to the left, the central, and the right component of $H_{\nu,\rho}$, respectively. Denote by $\sigma_{\mu_d}$ the spirality of the restriction of $H_{\nu,\rho}$ to $G_{\mu_d,\rho}$, with $d \in \{l,c,r\}$. By Lemma~\ref{le:spirality-P-node-3-children}, it suffices to show that $G_{\mu_{\max},\rho}$, $G_{\mu_{\md},\rho}$, and $G_{\mu_{\min},\rho}$ admit spiralities $\sigma_{\mu_l}$, $\sigma_{\mu_c}$, and $\sigma_{\mu_r}$, respectively.
		%$|\sigma_{\mu_l}| \in \Sigma_{\mu_{\max}}^+$, $|\sigma_{\mu_c}| \in \Sigma_{\mu_{\md}}^+$, and $|\sigma_{\mu_r}| \in \Sigma_{\mu_{\min}}^+$.
		Observe that, since $M > 4$, by Lemma~\ref{le:spirality-P-node-3-children} we have $\sigma_{\mu_d} > 0$.
		
		Let $d,d' \in \{l,c,r\}$ with $d \neq d'$ and let $M_{\mu_d}$ be the maximum value of spirality in $\Sigma^+_{\mu_d,\rho}$. We claim that if $\sigma_{\mu_{d'}} \leq M_{\mu_d}$ then $G_{\mu_d,\rho}$ admits spirality $\sigma_{\mu_{d'}}$: Since  $M>4$, by Lemma~\ref{le:spirality-P-node-3-children} we have $M_{\mu_d} \ge 3$ and if $\mu_d$ is jump-1, the claim holds by Theorem~\ref{th:spirality-sets}; if $\mu_d$ is not jump-1, by Lemma~\ref{le:spirality-P-node-3-children} it follows that $M$, $M_{\mu_d}$, and  $\sigma_{\mu_{d'}}$ have the same parity and, by Theorem~\ref{th:spirality-sets} the claim holds.
		
		We now show separately that: \textsf{(a)}~$G_{\mu_{\max},\rho}$ admits spirality $\sigma_{\mu_l}$, \textsf{(b)}~$G_{\mu_{\md},\rho}$ admits spirality $\sigma_{\mu_c}$, and \textsf{(c)} $G_{\mu_{\min},\rho}$ admits spirality $\sigma_{\mu_r}$.
		
		\smallskip\noindent \textsf{Proof of~(a):} Since by definition $M_{\mu_{\max}}\ge M_{\mu_l}\ge \sigma_{\mu_l}$, by the claim above we have that $G_{\mu_{\max},\rho}$ admits spirality $\sigma_{\mu_l}$.
		
		\smallskip\noindent  \textsf{Proof of~(b):} If $\mu_\md=\mu_c$ we are done. Else, suppose that $\mu_\md=\mu_l$. Since $\sigma_{\mu_l}\ge \sigma_{\mu_c}$, we have $M_\md=M_l\ge \sigma_{\mu_l}> \sigma_{\mu_c}$ and, consequently, by the claim $G_{\mu_{\md},\rho}$ admits spirality $\sigma_{\mu_c}$. Finally, suppose that $\mu_\md = \mu_r$. If $\mu_{\min}=\mu_c$ then $M_{\mu_{\md}} \ge M_{\mu_{\min}} \ge \sigma_{\mu_c}$; if $\mu_{\min}=\mu_l$ then $M_{\mu_{\md}} \ge M_{\mu_{\min}} \ge \sigma_{\mu_l} > \sigma_{\mu_c}$. Hence, by the claim,  $G_{\mu_{\md},\rho}$ admits spirality $\sigma_{\mu_c}$.
		
		\smallskip\noindent  \textsf{Proof of~(c):} Since $\sigma_{\mu_r}<\sigma_{\mu_c}<\sigma_{\mu_l}$, for any $\mu\in \{\mu_{\min}, \mu_{\md}, \mu_{\max}\}$, $\sigma_{\mu_r}\le M_{\mu}$. Hence, by the claim,  $G_{\mu_{\min},\rho}$ admits spirality $\sigma_{\mu_r}$.
	\end{claimproof}
	
	%\textcolor{red}{TO DO - Lemma 14+16+17 of Giacomo}
	%First show Lemma 14 of Giacomo - easy and short; then show how to compute the maximum (Lemma 16 of Giacomo) and then show how to compute the interval (Lemma 17 of Giacomo).
	
	By the claim, to compute $M$ when $M > 4$, we can restrict to consider only rectilinear planar representations of $G_{\nu,\rho}$ where $G_{\mu_{\max},\rho}$, $G_{\mu_{\md},\rho}$, and $G_{\mu_{\min},\rho}$ occur in this left-to-right order. Let $\overline{M}=\min\{M_{\mu_{\max}}-2,M_{\mu_{\md}}, M_{\mu_{\min}}+2\}$. By Lemma~\ref{le:spirality-P-node-3-children}, we have $M\le \overline{M}$. 
	We test in $O(1)$ time whether $\nu$ is jump-1; by Theorem~\ref{th:spirality-sets}, it is sufficient to check whether $G_{\nu,\rho}$ either admits both spiralities 0 and 1 or both spiralities 1 and 2. If $\nu$ is jump-1, by Lemma~\ref{le:spirality-P-node-3-children}, all the children of $\nu$ are jump-1. Hence, $\mu_{\max}$, $\mu_{\md}$, and $\mu_{\min}$ admit spiralities $\overline{M}-2$, $\overline{M}$, and $\overline{M}+2$, respectively, which implies that $M=\overline{M}$. Suppose vice versa that $\nu$ is not jump-1. In this case, we check in $O(1)$ if $G_{\nu,\rho}$ admits spirality $\overline{M}$. If so, $M=\overline{M}$. Otherwise, $M$ and $\overline{M}$ have opposite parity, which implies that $M$ and $\overline{M}-1$ have the same parity, and $M \le \overline{M}-1$. Since $\overline{M}-1=\min\{M_{\mu_{\max}}-2,M_{\mu_{\md}}, M_{\mu_{\min}}+2\}-1$, we have that $\mu_{\max}$,  $\mu_{\md}$, and $\mu_{\min}$ admit spiralities $(\overline{M}-2)-1$,   $\overline{M}-1$, and $(\overline{M}+2)-1$, respectively, i.e., $M=\overline{M}-1$.   
	
	\smallskip
	Based on $M$, we finally determine the structure of $\Sigma_{\nu,\rho}^+$ in $O(1)$ time. Namely, we check in $O(1)$ time if $\nu$ is jump-1; thanks to Theorem~\ref{th:spirality-sets} it suffices to check whether $G_{\nu,\rho}$ admits spiralities 0 and 1 or spiralities 1 and 2. Suppose that $\nu$ is jump-1; if it contains 0, then $\Sigma_{\nu,\rho}^+=[0,M]^1$; else $\Sigma_{\nu,\rho}^+=[1,2]^1$. Suppose vice versa that $\nu$ is not jump-1. If $M \leq 1$ then $\Sigma_{\nu,\rho}^+=[M]$. Otherwise, if $M$ is odd $\Sigma_{\nu,\rho}^+=[1,M]^2$ and if $M$ is even $\Sigma_{\nu,\rho}^+=[0,M]^2$.
	%	Based on $M$ we finally determine the structure of $\Sigma_{\nu,\rho}^+$ in $O(1)$ time. Suppose that $\nu$ is jump-1; if it contains 0, then $\Sigma_{\nu,\rho}^+=[0,M]^1$; else $\Sigma_{\nu,\rho}^+=[1,2]^1$. Suppose vice versa that $\nu$ is not jump-1. If $M \leq 1$ then $\Sigma_{\nu,\rho}^+=[M]$. Otherwise, if $M$ is odd $\Sigma_{\nu,\rho}^+=[1,M]^2$ and if $M$ is even $\Sigma_{\nu,\rho}^+=[0,M]^2$.
\end{proof}

\begin{restatable}{lemma}{leTimePTwo}\label{le:timeP2}
	Let $G$ be an \pisp and let $T_\rho$ be a rooted SPQ$^*$-tree of $G$. Let $\nu$ be a P-node of $T_{\rho}$ with two children and assume that, for each child $\mu$ of $\nu$ in $T_{\rho_i}$, the set $\Sigma^+_{\mu,\rho_i}$ is given.
	The set $\Sigma^+_{\nu,\rho}$ can be computed in $O(1)$ time.
\end{restatable}
\begin{proof}
	We follow the same proof strategy as for Lemma~\ref{le:timeP3}. By Lemma~\ref{le:spirality-P-node-2-children}, for any given integer $\sigma_\nu$, one can test in $O(1)$ time whether $G_{\nu,\rho}$ admits spirality $\sigma_\nu$. Indeed, it suffices to test whether there are four binary numbers $\alpha_u^l$, $\alpha_v^l$, $\alpha_u^r$, and $\alpha_v^r$ such that $1 \leq \alpha_u^l + \alpha_u^r \leq 2$, $1 \leq \alpha_v^l + \alpha_v^r \leq 2$, and for which one child of $\nu$ admits spirality $\sigma_\nu + \alpha_u^l + \alpha_v^l$ and the other child of $\nu$ admits spirality $\sigma_\nu - \alpha_u^r + \alpha_v^r$. Testing this condition requires a constant number of checks. 
	
	By Theorem~\ref{th:spirality-sets}, $G_{\nu,\rho}$ is rectilinear planar if and only if it admits spirality either $0$ or $1$. Based on the reasoning above, we can check this property in $O(1)$ time; if it does not hold, then $\Sigma^+_{\nu,\rho}=\emptyset$. Otherwise, we determine the maximum value $M$ in $\Sigma^+_{\nu,\rho}$. By Theorem~\ref{th:spirality-sets}, it suffices to find a value $\sigma_\nu$ such that $\nu$ admits spirality $\sigma_\nu$ but it does not admit spirality $\sigma_\nu + 1$ and $\sigma_\nu + 2$; if we find such a value, then $M=\sigma_\nu$. We prove how to find $M$ in $O(1)$ time.  
	
	For each $i=0,...,4$, we first check in $O(1)$ time whether $M=i$. If this is not the case, then $M>4$. To find $M$ in this case, we claim a property similar to the case of a P-node with three children. Denote by $\mu_{\max}$ a child of $\nu$ whose maximum value is not smaller than the other. Let $\mu_{\min}$ be the remaining child.
	
	%\clPTwo
	\begin{restatable}{claim}{clPTwo}
		Let $M$ be the maximum value in $\Sigma^+_{\nu,\rho}$. If $M > 4$, there exists a rectilinear planar representation of $G_{\nu,\rho}$ with spirality $M$ where $G_{\mu_{\max},\rho}$ and $G_{\mu_{\min},\rho}$ appear in this left-to-right order.
	\end{restatable}
	\begin{claimproof}
		Let $H_{\nu,\rho}$ be a rectilinear planar representation of $G_{\nu,\rho}$ with spirality~$M$. If $G_{\mu_{\max},\rho}$ is the left child in $H_{\nu,\rho}$, we are done. Otherwise we show that there exists a rectilinear planar representation $H'_{\nu,\rho}$ of $G_{\nu,\rho}$ with spirality $M$ such that $G_{\mu_{max},\rho}$ is the left child. Since $H_{\nu,\rho}$ has the maximum possible value of spirality, we have $\alpha_u^r=\alpha_v^r=1$. Since $\mu_{\min}$ is the left child in $H_{\nu,\rho}$ and $M>4$, by Lemma~\ref{le:spirality-P-node-2-children}, $\sigma_{\mu_{\min}}>2$. 
		By Property~(b) of Lemma~\ref{le:intervalSupportFirst}, there exists a rectilinear planar representation $H_{\mu_{\min},\rho}'$ of $G_{\mu_{\min},\rho}$ with spirality $\sigma_{\mu_{\min},\rho}'=\sigma_{\mu_{\min}}-2$. Also, by  Lemma~\ref{le:intervalSupportSecond} we can assume that $\sigma_{\mu_{\min}}-\sigma_{\mu_{\max}}=g$, with $2\le g\le3$. We have $M_{\mu_{\max}}\ge M_{\mu_{\min}} \ge \sigma_{\mu_{\min}}= g+\sigma_{\mu_{\max}}$. Hence, by Theorem~\ref{th:spirality-sets}, if $\sigma_{\mu_{\max}}$ and $M_{\mu_{\max}}$ have different parities, then $\mu_{\max}$ is jump-1, otherwise it is jump-2. In both cases, $\mu_{\max}$ admits spirality $\sigma_{\mu_{\max}}'=\sigma_{\mu_{\max}}+2$.
		We have $\sigma_{\mu_{\max}}'- \sigma_{\mu_{\min}}'=\sigma_{\mu_{\max}}+2-\sigma_{\mu_{\max}}-2=\sigma_{\mu_{\max}}-\sigma_{\mu_{\max}}=g$. Hence, by Lemma~\ref{le:spirality-P-node-2-children}, there exists a rectilinear planar representation $H'_{\nu,\rho}$ that contains $H_{\mu_{\min},\rho}'$ and $H_{\mu_{\max}\rho}'$ in this left-to-right order and such. The spirality $\sigma_\nu'$ of $H'_{\nu,\rho}$ is  $\sigma_\nu'=\sigma_{\mu_{\min}}'+2=\sigma_{\mu_{\min}}\ge \sigma_{\mu_{\max}}+2=M$.
	\end{claimproof}

	%		We first show that, for any integer $\sigma_\nu$, it is possible to check in $O(1)$ time if $\nu$ admits spirality $\sigma_\nu$. Let $h\in [0,2]$ and $\mu$ be a child of $\nu$. Let $\mu'$ be the other child of $\nu$. By Lemma~\ref{le:spirality-P-node-2-children} there exists a rectilinear planar representation of $G_{\nu,\rho}$ where $\mu$ is the left child and where $\alpha_u^l+\alpha_u^r=h$ if and only if there exists a rectilinear planar representation of $G_{\mu,\rho}$ and $G_{\mu',\rho}$ with spirality $\sigma_{\mu}$ and $\sigma_{\mu'}$, respectively, such that $\sigma_\nu=\sigma_{\mu}-h$ and  $\sigma_{\mu}-\sigma_{\mu'}\in [2,4]$. This test requires $O(1)$. Since $h$ can assume a constant number of values and since $\nu$ has two children, it implies that it is possible to test if $\nu$ admits spirality $\sigma_\nu$ in $O(1)$ time.
	%		
	When $M>4$, by Lemma~\ref{le:spirality-P-node-2-children}, we have $M_{\mu_{\min}}>2$ and $M_{\mu_{\max}}>2$. By the claim above we can restrict to consider only rectilinear planar representations of $G_{\nu,\rho}$ where $G_{\mu_{\max},\rho}$ and $G_{\mu_{\min},\rho}$ are the left and right child, respectively.   
	Also, we can restrict to consider $\alpha_u^r=\alpha_v^r=1$.
	By Lemma~\ref{le:spirality-P-node-2-children}, $M\le \min\{M_{\mu_{\max}}, M_{\mu_{\min}}+2\}$. 
	%We consider two cases: $M_{\mu_{\max}}\ge M_{\mu_{\min}}+2$ and $M_{\mu_{\max}}< M_{\mu_{\min}}+2$.
	
	Suppose first that $M_{\mu_{\max}}\ge M_{\mu_{\min}}+2$, which implies $M\le M_{\mu_{\min}}+2$. We show that in fact $M = M_{\mu_{\min}}+2$, i.e., $\nu$ admits spirality $M_{\mu_{\min}}+2$. Since $M_{\mu_{\max}}\ge M_{\mu_{\min}}+2$, we have that $\mu_{\max}$ admits spirality $M_{\mu_{\min}}+2$ or $M_{\mu_{\min}}+3$. Since $M_{\mu_{\max}} \ge M_{\mu_{\min}}+2$, we have that $\mu_{\max}$ admits spirality $M_{\mu_{\min}}+2$ or $M_{\mu_{\min}}+3$; this implies that we can realize a rectilinear planar representation of $G_{\nu,\rho}$ whose restrictions to $G_{\mu_{\min},\rho}$ and to $G_{\mu_{\max},\rho}$ have spiralities $\sigma_{\mu_{\min}} = M_{\mu_{\min}}$ and $\sigma_{\mu_{\max}}\in [M_{\mu_{\min}}+2,M_{\mu_{\min}}+3]$, respectively. By Lemma~\ref{le:spirality-P-node-2-children}, if $\sigma_{\mu_{\max}} = M_{\mu_{\min}}+2$ then $G_{\nu,\rho}$ admits spirality $M_{\mu_{\min}}+2$ for $\alpha_u^l=\alpha_v^l=0$. If $\sigma_{\mu_{\max}} = M_{\mu_{\min}}+3$ then $G_{\nu,\rho}$ admits spirality $M_{\mu_{\min}}+2$ for $\alpha_u^l=1$ and $\alpha_v^l=0$.        
	
	\smallskip
	Suppose vice versa that $M_{\mu_{\max}}< M_{\mu_{\min}}+2$, which implies  $M\le M_{\mu_{\max}}$. In this case we show that either $M=M_{\mu_{\max}}$ or $M=M_{\mu_{\max}}-1$. Since $M_{\mu_{\min}} > M_{\mu_{\max}}-2$, we have that $\mu_{\min}$ admits spirality $M_{\mu_{\max}}-2$ or $M_{\mu_{\max}}-3$. If $\mu_{\min}$ admits spirality $M_{\mu_{\max}}-2$, then we can realize a rectilinear planar representation of $G_{\nu,\rho}$ whose restrictions to 
	$G_{\mu_{\min},\rho}$ and to $G_{\mu_{\max},\rho}$ have spiralities $\sigma_{\mu_{\max}} = M_{\mu_{\max}}$ and  $\sigma_{\mu_{\min}} = M_{\mu_{\max}}-2$, respectively. By Lemma~\ref{le:spirality-P-node-2-children}, this representation has spirality $M_{\mu_{\max}}$, which implies $M = M_{\mu_{\max}}$. If $\mu_{\min}$ does not admit spirality $M_{\mu_{\max}}-2$, it admits spirality $M_{\mu_{\max}}-3$ and $M\le M_{\mu_{\max}}-1$. In this case we realize a rectilinear representation of $G_{\nu,\rho}$ whose restrictions to 
	$G_{\mu_{\min},\rho}$ and to $G_{\mu_{\max},\rho}$ have spiralities $\sigma_{\mu_{\max}} = M_{\mu_{\max}}$ and  $\sigma_{\mu_{\min}} = M_{\mu_{\max}}-3$. By Lemma~\ref{le:spirality-P-node-2-children}, this representation has spirality $M_{\mu_{\max}}-1$, which implies that $M=M_{\mu_{\max}}-1$.  		
	
	Based on $M$, we finally determine the structure of $\Sigma_{\mu,\rho}^+$ in $O(1)$ time. We have that $G_\nu$ admits a rectilinear planar representation with spirality $M$ and $M-1$. Indeed, if $\nu$ has spirality $M$, then $\alpha_u^r = \alpha_v^r = 1$ and, by Lemma~\ref{le:intervalSupportSecond}, at least one of $\alpha_u^l$ and $\alpha_v^l$ equals 0, say for example $\alpha_u^l=0$. For $\alpha_u^l=1$ we get spirality $M-1$. Hence, by Theorem~\ref{th:spirality-sets}, $\nu$ is jump-1:  If $M=2$ and $G_\nu$ does not admit a representation with spirality $0$, $\Sigma_{\mu,\rho}^+=[1,2]^1$. Otherwise, $\Sigma_{\mu,\rho}^+=[0,M]^1$.
\end{proof}

\smallskip
\myparagraph{Testing algorithm.} 
Let $G$ be a \pisp different from a simple cycle, $T_\rho$ be a rooted SPQ$^*$-tree of $G$, and $\nu$ be the child of $\rho$ in $T_\rho$.  
At the level of the root the test consists of verifying whether $\nu$ and $\rho$ admit spirality values $\sigma_\nu \in \Sigma^+_{\nu,\rho}$ and $\sigma_\rho \in \Sigma^+_{\rho,\rho}$, respectively, such that $\sigma_\nu + \sigma_\rho = 4$. 

\begin{restatable}{lemma}{leTimeRoot}\label{le:timeRoot}
	Let $G$ be an \pisp, $T_\rho$ be a rooted SPQ$^*$-tree of $G$, and $\nu$ be the child of $\rho$ in $T_\rho$. If $G_{\nu,\rho}$ is rectilinear planar, one can test whether $G$ is rectilinear planar in $O(1)$ time.
\end{restatable}
\begin{proof}
	The child $\nu$ of the root $\rho$ in $T_\rho$ is either a P-node or an S-node with poles $u$ and $v$ of indegree two. In both cases, the alias vertices associated with $u$ and $v$ are dummy vertices that subdivide the reference chain represented by $\rho$.
	Let $\ell$ be the length of the reference chain represented by $\rho$; we know that $\Sigma^+_{\rho,\rho}=[0,\ell-1]^1$. In every cycle $C$ of a rectilinear representation of $G$, the difference between the number of right turns and the number left turns, while walking along the boundary of $C$ clockwise, is equal to four. Hence, $G$ is rectilinear planar if and only if we can find a value $\sigma_\nu \in \Sigma^+_{\nu,\rho}$ and a value $0 \leq k \leq \ell-1$ such that $\sigma_\nu + k = 4$. Hence, for any pair $\sigma_\nu \in \{0,1,2,3,4\}$ and $k \in \{0,1,2,3,4\}$ such that $\sigma_\nu + k = 4$ we just test whether  $\sigma_\nu \in \Sigma^+_{\nu,\rho}$ and $k \in \{0, \dots, \ell-1\}$. This requires a constant number of checks.
\end{proof}

%\myparagraph{Proof of Lemma~\ref{le:timeP}}.

%\myparagraph{Summarizing result.}

%Lemmas~\ref{le:P3Interval} and~\ref{le:P2Interval} directly imply the following.

%\begin{restatable}{lemma}{letimeP}\label{le:timeP}
%	Let $G$ be an \pisp and let $T_\rho$ be the SPQ$^*$-tree of $G$ rooted at a Q$^*$ node $\rho$.
%	Let $\nu$ be a P-node of $T_{\rho}$ and assume that, for each child $\mu$ of $\nu$ in $T_{\rho_i}$, $G_{\mu,\rho_i}$ is rectilinear planar and $\Sigma^+_{\mu,\rho_i}$ is given.
%	There exists an $O(1)$-time algorithm that tests whether $G_{\nu,\rho_i}$ is rectilinear planar and, if so, it computes $\Sigma^+_{\nu,\rho}$.
%\end{restatable}

%\begin{restatable}{lemma}{letimeP3}\label{le:timeP2}
%	Let $G$ be an \pisp and let $T_\rho$ be the SPQ$^*$-tree of $G$ rooted at a Q$^*$ node $\rho$.
%	Let $\nu$ be a P-node of $T_{\rho}$ with three children and assume that $\Sigma^+_{\mu,\rho}$ is given for every child $\mu$ of $\nu$. There exists an algorithm that computes $\Sigma^+_{\nu,\rho}$
%	in $O(1)$ time.
%\end{restatable}

\begin{restatable}{theorem}{thTimeTest}\label{th:timeTest}
	Let $G$ be an \pisp with $n$ vertices. There exists an $O(n)$-time algorithm that tests whether $G$ is rectilinear planar.
\end{restatable}
\begin{proof}
	If $G$ is a simple cycle, the test is trivial, as~$G$ is rectilinear planar if and only if it contains at least four vertices. Assume that~$G$ is not a simple cycle. Let~$T$ be the SPQ$^*$-tree of~$G$, and let
	$\rho_1, \dots, \rho_h$ be the Q$^*$-nodes of~$T$. For each $i=1, \dots, h$, the testing algorithm performs a post-order visit of~$T_{\rho_i}$. During this visit, for every non-root node $\nu$ of $T_{\rho_i}$ the algorithm computes~$\Sigma^+_{\nu,\rho_i}$ by using Lemmas~\ref{le:timeQ},~\ref{le:timeS}, and~\ref{le:timeP}. If $\Sigma^+_{\nu,\rho_i}=\emptyset$, the algorithm stops the visit, discards~$T_{\rho_i}$, and starts visiting~$T_{\rho_{i+1}}$ (if $i < h$). If the algorithm achieves the root child $\nu$ and if $\Sigma^+_{\nu,\rho_i} \neq \emptyset$, it checks whether $G$ is rectilinear planar by using Lemma~\ref{le:timeRoot}: if so, the test is positive and the algorithm does not visit the remaining trees; otherwise it discards~$T_{\rho_i}$ and starts visiting~$T_{\rho_{i+1}}$~(if~$i < h$).
	
	We now analyze the time complexity of the testing algorithm. When the algorithm visits $T_{\rho_1}$, for any $\nu$ with $n_\nu$ children it computes $\Sigma^+_{\nu,\rho_1}$ in $O(n_\nu)$ time if $\nu$ is an S-node (Lemma~\ref{le:timeS}) and in $O(1)$ time otherwise (Lemmas~\ref{le:timeQ} and~\ref{le:timeP}). Hence, since $T$ has $O(n)$ nodes and the sum of the degree of its nodes is $O(n)$, $T_{\rho_1}$ is visited in $O(n)$ time. For every other  visit of a tree $T_{\rho_i}$ $(i=2, \dots, h)$, the algorithm spends $O(1)$ to compute the non-negative spirality set of each node~$\nu$, even when~$\nu$ is an S-node (Lemma~\ref{le:timeS}). Also, the algorithm does not need to recompute $\Sigma^+_{\nu,\rho_i}$ if the parent of $\nu$ in $T_{\rho_i}$ coincides with the parent of $T_{\rho_j}$ for some $j \in \{1, \dots, i-1\}$; in fact, in this case, $\Sigma^+_{\nu,\rho_i}$ coincides with $\Sigma^+_{\nu,\rho_j}$, which was previously computed and can be simply reused. Therefore, since every node $\nu$ with $n_\nu$ children changes its parent node $n_\nu$ times over all visits of $T_{\rho_i}$ ($i=2, \dots, h$), and since the sum of the degree of the nodes of $T$ is $O(n)$, the algorithm needs to compute $O(n)$ non-negative spirality sets in total, each requiring $O(1)$ time. Thus, the testing algorithm takes $O(n)$ time.
\end{proof}

%%%%% CONSTRUCTION
%\section{Computing a Rectilinear Planar Representation}\label{se:construction}
%Use a top-down visit technique that is similar to the one of last year
\begin{figure}[t]
	\centering
	\includegraphics[width=0.8\columnwidth,page=3]{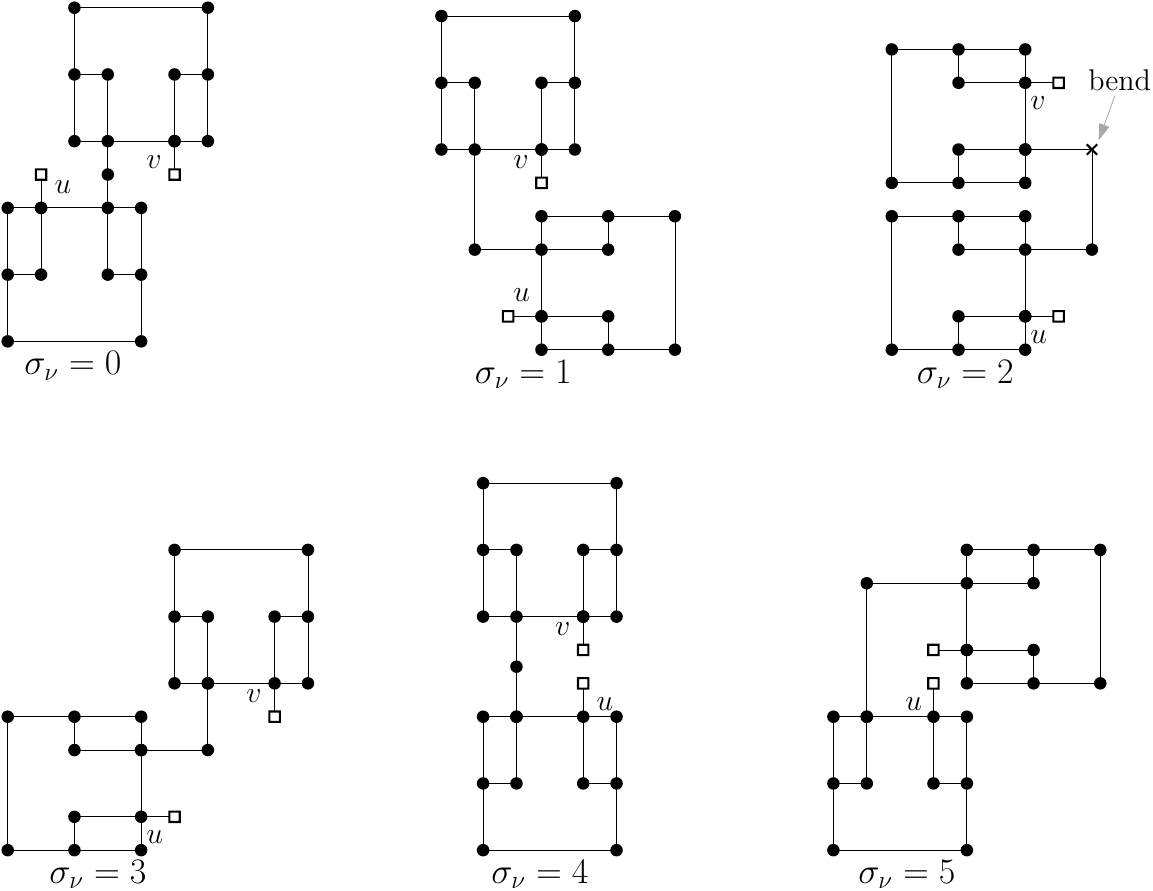}
	%\caption{A component $G_{\nu,\rho}$ such that $\Sigma^+_{\nu,\rho} = \{0,1,3,4,5\}$. A bend is drawn as $\times$.}
	\caption{Component that admits spiralities 0,1,3,4,5. Spirality 2 needs a bend~($\times$).}
	\label{fi:irregular}
\end{figure}
%%%%% FINAL
\section{Final Remarks}\label{se:final}
%As we are going to show in Section~\ref{se:final}, such a regularity does not hold for general SP-graphs.
In this paper we showed that rectilinear planarity for \pisps can be tested in linear time in the variable embedding setting. It is not difficult to see that if the testing is positive, a rectilinear planar representation of~$G$ can also be constructed in linear time by a variant of the technique described in~\cite{DBLP:conf/gd/Didimo0LO20}: Visit $T_\rho$ top-down and for each node $\nu$ compute a target value of spirality in $\Sigma^+_{\nu,\rho}$, based on whether $\nu$ is an S-node, a P-node, or a Q$^*$-node.

The problem about whether the result of Theorem~\ref{th:timeTest} can be extended to every SP-graph remains open. We just observe here that the spirality set of a component of an SP-graph that is not independent-parallel may not exhibit a behavior like the one described in Theorem~\ref{th:spirality-sets}. For example, the component shown in Fig.~\ref{fi:irregular} is rectilinear planar for all spirality values from 0 to 5 except 2.

%\newpage

\bibliography{bibliography}
\bibliographystyle{abbrvurl}

\clearpage
%\appendix
%\section*{\LARGE Appendix}\label{se:appendix}\vspace{1em}
%\makeatletter
%\noindent
%\rlap{\color[rgb]{0.51,0.50,0.52}\vrule\@width\textwidth\@height1\p@}%
%\hspace*{7mm}\fboxsep1.5mm\colorbox[rgb]{1,1,1}{\raisebox{-0.4ex}{%
%		\large\selectfont\sffamily\bfseries Appendix}}%
%\makeatother

%\section{Additional Material for Section~\ref{se:preli}}\label{app:preli}
%\input{app-preli.tex}

%\section{Additional Material for Section~\ref{se:spirality-pisp}}\label{app:spirality-pisp}
%\input{app-spirality.tex}

%\section{Additional Material for Section~\ref{se:intervals}}\label{app:intervals}
%\input{app-intervals.tex}

%\section{Additional Material for Section~\ref{se:testing}}\label{app:testing}
%\input{app-testing.tex}

%\section{Additional Material for Section~\ref{se:construction}}
%\input{app-construction.tex}

%\section{Additional Material for Section~\ref{se:final}}
%\input{app-final.tex}

\end{document}